\documentclass[a4paper, amsfonts, amssymb, amsmath, reprint, showkeys, nofootinbib, twoside,notitlepage]{revtex4-1}

\def\apj{{ApJ}}                 
\def\apss{{Ap\&SS}}             
\def\aap{{A\&A}}                
\def\aaps{{A\&AS}}              

\def\mnras{{MNRAS}}             
\def\prd{{Phys.~Rev.~D}}        
\usepackage{amsmath,amstext}
\usepackage[T1]{fontenc}
\usepackage{amssymb}
\usepackage{graphicx}
\usepackage{ae,aecompl}

\DeclareFontFamily{OT1}{pzc}{}
\DeclareFontShape{OT1}{pzc}{m}{it}{<-> s * [1.10] pzcmi7t}{}
\DeclareMathAlphabet{\mathpzc}{OT1}{pzc}{m}{it}

\usepackage{hyperref}
\usepackage{amsmath}
\usepackage{amssymb}
\usepackage{mathtools}
\usepackage{bm}
\usepackage{cleveref}
\usepackage{tensor}
\usepackage{braket}
\usepackage{enumitem}
\usepackage{mhchem}
\usepackage{amsthm}
\usepackage{nccmath}
\usepackage{mathrsfs}
\usepackage{color}

\newcommand{\spc}{\quad \quad \quad}

\newcommand{\K}{{\mathcal{K}}}

\def\be{\begin{equation}}
\def\ee{\end{equation}}
\def\beq{\begin{eqnarray}}
\def\eeq{\end{eqnarray}}

\theoremstyle{definition}

\theoremstyle{theorem}
\newtheorem{theorem}{Theorem}

\begin{document}
\title{Thermoelectric conduction in General Relativity: a causal, stable, and well-posed theory}
\author{L.~Gavassino}
\affiliation{Department of Applied Mathematics and Theoretical Physics, University of Cambridge, Wilberforce Road, Cambridge CB3 0WA, United Kingdom}

\begin{abstract}
We present a covariantly stable first-order framework for describing charge and heat transport in isotropic rigid media embedded in curved spacetime. Working in the Lorenz gauge, we show that the associated initial value problem is both causal and locally well-posed in the fully nonlinear regime. We then apply such framework to explore a range of gravitothermoelectric effects in metals undergoing relativistic acceleration. These include (1) the separation of charge through acceleration, (2) the non-uniformity of Joule heating across accelerating circuits due to time dilation, and (3) the effect of redshift on magnetic diffusion. As an astrophysical application, we derive a relativistic Thomas-Fermi equation governing the charge distribution inside a compact object, also accounting for Seebeck charge displacements driven by cooling.
\end{abstract}

\maketitle
\noindent\textit{\bf Introduction -} Suppose we place an electric circuit just above the event horizon of a black hole, or on the floor of a rocket undergoing extremely large uniform acceleration. How does gravity/acceleration modify the flow of electric current? 

This simple question is the starting point of a long-standing line of research on \textit{gravitoelectric effects} \cite{TolmanStewart1916,DeWittDrag1966,Anandan1984,Ahmedov:1998mif,Ahmedov:1999bqr,Ahmedov:2002bx,Ummarino:2017bvz,WangStewartTolman2023,BeiPhD,Gavassino:2025wcy}, namely phenomena that cause an electric circuit to effectively function as a gravimeter/accelerometer. The most elementary example is the well-known result that, in stationary conditions, it is the redshifted current $I\sqrt{-g_{tt}}$ (rather than $I$ itself) that remains uniform along a wire \cite{Anandan1984}; see figure \ref{fig:redshifcurrent} for a quick proof. Another classic instance is the Schiff-Barnhill effect \cite{ShiffBarhill1966}, according to which a conductor in a gravitational field develops an equilibrium charge separation (the Stewart-Tolman effect \cite{TolmanStewart1916,Bhatia1969,Gavassino:2025wcy} is the same phenomenon with accelerations). While such effects are expected to be relevant primarily in locally charged compact object \cite{Bhatia1969,Bekenstein1971,Zhand1982,Alcock1986,Usov:2004iz,Usov:2004kj,Ray:2003gt,Ray:2004yr,Negreiros:2009fd,Yousaf2016,Carvalho:2018kbt,FRocha2020,Bhatti2020,Arbanil2022,Sharif:2023vpb,Masa2023}, some are sufficiently pronounced to be detectable even in Earth’s weak gravitational field, and have been exploited as precision tests of Einstein’s equivalence principle \cite{Jain1987}.

Despite the broad interest, a fully relativistic generalization of the standard solid-state theory of thermoelectricity (see e.g. \cite[\S 26]{landau8} or \cite[\S 21.C]{SommerfeldStatmech}) is still missing. Certain relativistic fluid theories do contain the relevant transport mechanisms \cite{Hernandez:2017mch,Fotakis:2019nbq,Dommes:2020ktk,Dash:2022xkz,GavassinoShokryMHD:2023qnw}, and we may just take their rigid limit. However, the resulting frameworks are either acausal (and therefore unstable \cite{Hiscock_Insatibility_first_order,Kost2000,GavassinoLyapunov_2020,GavassinoSuperlum2021}) or introduce a large number of additional ``second order'' transport coefficients whose values are generally unknown for relativistic materials. More importantly, none of the available models is known to admit a well-posed initial value formulation. In other words, it remains unclear whether their equations are mathematically solvable at all, particularly for objects undergoing relativistic rotation \cite{GavassinoAntonelli:2025umq}.

Here, we adopt the Benfica-Disconzi-Noronha-Kovtun (BDNK) strategy \cite{Bemfica2019_conformal1,Kovtun2019,BemficaDNDefinitivo2020} to construct, for the first time, a causal, stable, and well-posed theory for thermoelectricity in relativistic media, which reduces to the standard ``textbook'' framework \cite[\S 26]{landau8}\cite[\S 21.C]{SommerfeldStatmech} in the Newtonian limit. We then illustrate its predictive power by discussing several applications in simplified geometries.

\noindent\textit{\bf Conventions -} The metric is treated as a fixed background, and has signature $(-,+,+,+)$. We work in natural units, with $c\,{=}\,\hbar\,{=}\,k_B\,{=}\,1$, and $e^2{=}\,4\pi/137$ \cite[\S 58]{Srednicki_2007}. 

\begin{figure}
    \centering
    \includegraphics[width=0.6\linewidth]{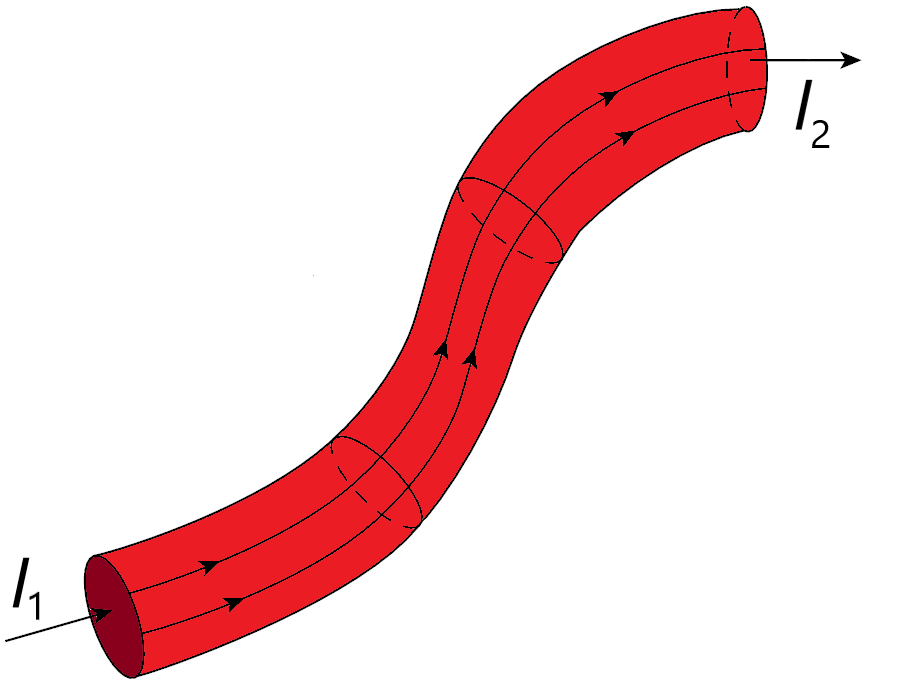}
\caption{Continuity of the redshifted current. Under stationary conditions, the Lie derivative $\mathfrak{L}_{\partial_t} J$ vanishes, where $J^\nu$ is the current density, and $\partial_t$ is a Killing vector. Combining this with the conservation law $\nabla_\mu J^\mu{=}0$ and the Killing property $\nabla_\mu (\partial_t)^\mu{=}0$, we find that $\nabla_\mu [(\partial_t)^\mu J^\nu{-}J^\mu (\partial_t)^\nu]=0$. Integrating over the portion of wire in the figure, and applying Stokes' theorem, we obtain $\int_{\text{Area}_1} J^\mu(\partial_t)^\nu dS_{\mu \nu}=\int_{\text{Area}_2} J^\mu(\partial_t)^\nu dS_{\mu \nu}$, or, equivalently, $I_1 ||\partial_t||_1=I_2 ||\partial_t||_2$.}
    \label{fig:redshifcurrent}
\end{figure}

\noindent\textit{\bf Derivation of the theory -} Following \cite{Anandan1984}, consider a conductive medium, whose positive ions are constrained to move with a macroscopic velocity field $u^\mu$ that is Born-rigid, i.e. proportional to a timelike Killing vector $\K^\mu\,{=}\,\K u^\mu$ \cite{Born1909,Herglotz1910,Noether1910}. In thermodynamic equilibrium, all fields, including the Faraday tensor $F_{\mu\nu}{=}\nabla_\mu A_\nu{-}\nabla_\nu A_\mu$, are stationary, i.e. $\mathfrak{L}_\K \text{``fields''}\,{=}\,0$, and we have that \cite{Hernandez:2017mch}
\vspace{-0.1cm}
\begin{equation}\label{gringoo}
\nabla_\mu (\K T)=0 \, , \spc \nabla_\mu (\K\mu)-\K\mathcal{E}_\mu=0\, , 
\end{equation}
where $T$ is the temperature, $\mu$ is the electric charge chemical potential, and $\mathcal{E}_\mu{=}F_{\mu \nu}u^\nu$ the electric field in the conductor's local rest frame \cite[\S 17.2.3]{special_in_gen}. Moreover, despite the possible presence of charge separation, there is no current or heat flow in equilibrium \cite{GavassinoShokryMHD:2023qnw,Gavassino:2025wcy}, i.e.
\vspace{-0.1cm}
\begin{equation}
\begin{split}
J^\mu ={}& \rho_{\text{lte}}(T,\mu;\text{``Ion configuration''})\, u^\mu\, , \\
W^\mu ={}&   \varepsilon_{\text{lte}}(T,\mu;\text{``Ion configuration''})\, \K^\mu \, ,\\
\end{split}
\end{equation}
where ``lte'' stands for local thermodynamic equilibrium. Here, $J^\mu$ is the electric current, and $W^\mu =-T_\text{mat}^{\mu \nu}\K_\nu$, with $T_\text{mat}^{\mu \nu}$ is the matter stress-energy tensor \cite{GavassinoAntonelli:2025umq}.


\newpage
Now, suppose the medium is \textit{not} in thermodynamic equilibrium (but the ions still move rigidly). Then, we can model corrections to $J^\mu$ and $W^\mu$ as linear responses to the left sides of \eqref{gringoo} not being zero, with susceptibilities $\sigma_i,\kappa_i\,{=}\,\mathcal{O}(\tau)$ ($\tau=\,$``electron mean free path''). Then, the most general isotropic first-order theory is
\vspace{-0.1cm}
\begin{equation}\label{original}
\begin{split}
\K J^\mu {=}\rho_{\text{lte}} \K^\mu & {-}(\sigma_1 g^{\mu \nu}{+}\sigma_2 u^\mu u^\nu) [\nabla_\nu (\K\mu){-}\K\mathcal{E}_\nu] \\
& {-}(\sigma_3 g^{\mu \nu}{+}\sigma_4 u^\mu u^\nu) \nabla_\nu (\K T) +\mathcal{O}(\tau^2) , \\
W^\mu {=}\varepsilon_{\text{lte}} \K^\mu & {-}(\kappa_1 g^{\mu \nu}{+}\kappa_2 u^\mu u^\nu) [\nabla_\nu (\K\mu){-}\K \mathcal{E}_\nu] \\
& {-}(\kappa_3 g^{\mu \nu}{+}\kappa_4 u^\mu u^\nu) \nabla_\nu (\K T) +\mathcal{O}(\tau^2). \\
\end{split}
\end{equation}
But let us recall that, out of equilibrium, $T$ and $\mu$ are not uniquely defined, as the electrons' statistical distribution receives $\mathcal{O}(\tau)$-corrections, and there are many ways to approximate its shape with a Maxwellian fit \cite{KovtunTemperature2022vas,Weinberg1971,GavassinoInfiniteOrded2024pgl}. This allows us to introduce field redefinitions of the form \cite{GavassinoAntonelli:2025umq}
\vspace{-0.1cm}
\begin{equation}\label{hydrframe}
\begin{split}
T={}& \Tilde{T}+\tau_1 u^\nu \nabla_\nu \Tilde{T} +\tau_2 u^\nu \nabla_\nu \Tilde{\mu}+\mathcal{O}(\tau^2) \, , \\
\mu={}& \Tilde{\mu}+\tau_3 u^\nu \nabla_\nu \Tilde{T} +\tau_4 u^\nu \nabla_\nu \Tilde{\mu}+\mathcal{O}(\tau^2) \, , \\
\end{split}
\end{equation}
for some $\tau_i {=}\mathcal{O}(\tau)$. With a field-redefinition of this kind, we can always set $\sigma_2{=}\sigma_4{=}\kappa_2{=}\kappa_4{=}0$, as there are 4 fixable $\tau_i$'s in \eqref{hydrframe}. Neglecting the $\mathcal{O}(\tau^2)$, we finally obtain
\vspace{-0.1cm}
\begin{equation}\label{themain}
\begin{split}
\K J^\mu ={}& \rho_{\text{lte}} \K^\mu  -\sigma_1 [\nabla^\mu (\K\mu){-}\K\mathcal{E}^\mu] -\sigma_3 \nabla^\mu (\K T) \, , \\
W^\mu ={}& \varepsilon_{\text{lte}} \K^\mu  -\kappa_1 [\nabla^\mu (\K\mu){-}\K \mathcal{E}^\mu] -\kappa_3  \nabla^\mu (\K T) \, . \\
\end{split}
\end{equation}
In the stationary limit, \eqref{themain} reduces to the theory of \cite{Anandan1984}, which allows one to assign physical interpretations to the transport coefficients (e.g. $\sigma_1 =$ ``electric conductivity'').

\noindent\textit{\bf Causality and well-posedness -} The vanishing of $\{\sigma_{2i},\kappa_{2i}\}$ allows us to establish several self-consistency results.
\vspace{-0.1cm}
\begin{theorem}
Suppose that \eqref{themain} holds, where $\sigma_i(T,\mu,x^\alpha)$ and $\kappa_i(T,\mu,x^\alpha)$ are such that $\sigma_1\kappa_3{-}\sigma_3 \kappa_1{\neq} 0$ everywhere. Then, in the Lorenz gauge \textup{(}$\nabla_\mu A^\mu \,{=}\, 0$\textup{)}, Maxwell's system
\vspace{-0.2cm}
\begin{equation}\label{maxwell}
\begin{split}
& \nabla_\mu F^{\nu \mu} =J^\nu \, , \\
& \nabla_\mu J^{\mu} =0 \, , \\
& \nabla_\mu W^{\mu} =\K \mathcal{E}_\mu J^\mu \, , \\
\end{split}
\end{equation} 
defines a well-posed Cauchy problem for the variables $\Psi{=}\{A^\nu,\mu,T\}$, given initial data on spacelike surfaces. Furthermore, the resulting dynamics propagate signals precisely at the speed of light.
\end{theorem}
\begin{proof}
It will be enough to show that the system \eqref{maxwell} takes the form $\nabla_\mu \nabla^\mu \Psi =\mathcal{F}(\Psi,\nabla_\alpha\Psi,x^\alpha)$, as then we can invoke Theorem 10.1.3 of \cite{Wald}. To this end, we recall that, in the Lorenz gauge, the first line of \eqref{maxwell} reads $\nabla_\mu \nabla^\mu A^\nu =R^\nu _\mu A^\mu {-}J^\nu$, where $R_{\mu \nu}$ is the Ricci curvature \cite[\S 22.4]{MTW_book}. This already has the form we want, since $J^\mu$ only contains first derivatives. The second and third lines can be rearranged, with a bit of algebra, in the form
\vspace{-0.1cm}
\begin{equation}\label{complichiamo}
\begin{bmatrix}
\sigma_1 & \sigma_3 \\
\kappa_1 & \kappa_3 \\
\end{bmatrix}
\begin{pmatrix}
\nabla_\mu \nabla^\mu \mu \\
\nabla_\mu \nabla^\mu T \\
\end{pmatrix}
= \text{``terms with at most }\nabla\Psi\text{''} \, .
\end{equation}
To arrive here, one needs to use the product rule, and the identity $\nabla_\mu \mathcal{E}^\mu=-u_\nu J^\nu +F^{\mu \nu}\nabla_\mu u_\nu$. Now, the condition $\sigma_1\kappa_3-\sigma_3 \kappa_1\neq 0$ guarantees that the matrix in \eqref{complichiamo} is invertible everywhere, allowing us to isolate $\nabla_\mu \nabla^\mu \Psi$.
\end{proof}

\newpage
\noindent\textit{\bf Entropy -} To first order in the mean free path, the entropy current of the conductor reads \cite{Israel1981}
\begin{equation}\label{entroycurrent}
\begin{split}
\K s^\mu = s_\text{lte} \K^\mu &{-}\dfrac{\kappa_1{-}\mu\sigma_1}{T} [\nabla^\mu (\K\mu){-}\K\mathcal{E}^\mu] \\ &{-}\dfrac{\kappa_3{-}\mu\sigma_3}{T} \nabla^\mu (\K T)+\mathcal{O}(\tau^2) \, .\\
\end{split}
\end{equation}
Its four-divergence can be computed with the aid of \eqref{maxwell}. Defined
$V^\mu =[\K^{-1}\nabla^\mu(\K \mu){-}\mathcal{E}^\mu, \K^{-1}\nabla^\mu(\K T)]$, we have
\begin{equation}\label{secondlaw}
T\nabla_\mu s^\mu =V^\mu
\begin{bmatrix}
\sigma_1 & \sigma_3 \\
\\
\dfrac{\kappa_1{-}\mu\sigma_1}{T} & \dfrac{\kappa_3{-}\mu\sigma_3}{T} \\
\end{bmatrix}V^T_\mu +\mathcal{O}(\tau^2) \, .
\end{equation}
Then, Onsager's principle \cite[\S 14.5]{Callen_book} immediately tells us that $T\sigma_3=\kappa_1{-}\mu\sigma_1$, in agreement with \cite{Bajec:2025dqm}. Moreover, the second law of thermodynamics requires that the $2{\times}2$ (symmetric) matrix above be positive definite. In fact, the conservation laws in \eqref{maxwell} imply that $\K_\mu V^\mu {=}\mathcal{O}(\tau)$ onshell, so $V^\mu$ is effectively spacelike in \eqref{secondlaw}, up to an error that scales like $\mathcal{O}(\tau^2)$. The resulting inequalities are
\begin{equation}
\sigma_1>0 \, , \spc \kappa_3-\mu\sigma_3-T\sigma_3^2/\sigma_1>0 \, ,
\end{equation}
the latter being interpretable as the positivity of the heat conductivity \cite[\S 26]{landau8}.

\noindent\textit{\bf Covariant stability -} Since the theory is causal, stability in the rest frame implies covariant stability \cite{GavassinoSuperlum2021,GavassinoBounds2023myj}. Hence, we linearize \eqref{maxwell} around a uniform charge-neutral state in Minkowski space, with $\K{=}1$ and $u^\mu{=}(1,0,0,0)$. By analysing the ``curl'' Maxwell equations, we find that there are 4 transversal quasinormal modes, all of which obey a causal magnetic diffusion equation:
\begin{equation}\label{magneticdiffusion}
\partial_t \mathcal{B}_j =\dfrac{1}{\sigma_1}\partial_\mu \partial^\mu \mathcal{B}_j \, ,
\end{equation}
which is stable \cite{GavassinoBounds2023myj}. The remaining 4 quasinormal modes are longitudinal and have no magnetic field. 
Their evolution can be studied by linearizing $\nabla_\mu J^\mu=0$ and \eqref{secondlaw}. This results in the system
\begin{equation}\label{longtudo}
\left\{ \!
\begin{bmatrix}
\chi_{\mu \mu} & \chi_{\mu T} \\
\chi_{T\mu} & \chi_{TT} \\
\end{bmatrix}\!\partial_t{-}\!
\begin{bmatrix}
\sigma_1 & \sigma_3 \\
\sigma_3 & \sigma_5 \\
\end{bmatrix}\! \partial_\mu \partial^\mu 
\right\}\!
\begin{bmatrix}
\delta \mu \\
\delta T \\
\end{bmatrix}
{+}
\begin{bmatrix}
\sigma_1\\
\sigma_3\\
\end{bmatrix} \!
\partial_\mu \mathcal{E}^\mu
{=}\,0,
\end{equation}
where $\chi_{ij}$ is the susceptibility matrix $\partial(\rho,s)/\partial(\mu,T)$, and we defined $T\sigma_5=\kappa_3{-}\mu\sigma_3$. Let us note that the two matrices in \eqref{longtudo} are both symmetric and positive definite, and thus can be rewritten as $N\Lambda N^T$ and $NN^T$ respectively, with $N$ invertible and $\Lambda=\text{diag}(\lambda_1,\lambda_2)$, $\lambda_i>0$. Hence, assuming a plane-wave form $\Psi \propto e^{\Gamma t+ikz}$ ($k\in \mathbb{R}$), and invoking the equation $\partial_\mu \mathcal{E}^\mu=\delta J^0$, \eqref{longtudo} becomes
\begin{equation}\label{Ups}
\begin{split}
&[\Gamma^2+(\Pi{+}\Lambda) \Gamma +\Pi \Lambda+k^2]\Upsilon =0\, ,\\
\text{with }\quad &\Upsilon= N^T \begin{bmatrix}\!
\delta \mu \\
\delta T \\
\end{bmatrix}\, , \quad \Pi= N^T\begin{bmatrix}
1 & 0 \\
0 & 0 \\
\end{bmatrix}
N \, .\\
\end{split}
\end{equation}
In the Supplementary Material, we prove that, as long as $\Pi$ and $\Lambda$ are non-negative definite (which they are), the solutions of \eqref{Ups} always have $\mathfrak{Re}\Gamma\leq 0$, i.e. are stable. 

The regime of validity and physical content of each mode is discussed in the Supplementary Material.


\newpage
\noindent\textit{\bf Setup for applications 1,2,3 -} In what follows, we apply our theory to investigate extreme gravitoelectric phenomena arising at high accelerations. To this end, we discuss the behavior of rigid pieces of equipment aboard a rocket undergoing hyperbolic (i.e. uniformly accelerated) motion in Minkowski space, see Figure \ref{fig:rocket}. As usual, we work in Rindler coordinates ($ds^2{=}-g^2z^2 dt^2{+}dx^2{+}dy^2{+}dz^2$), so that the local velocity of the equipment aligns with the Killing vector $\partial_t$, and thus $\K =gz$.
\begin{figure}[t!]
    \centering
    \includegraphics[width=0.8\linewidth]{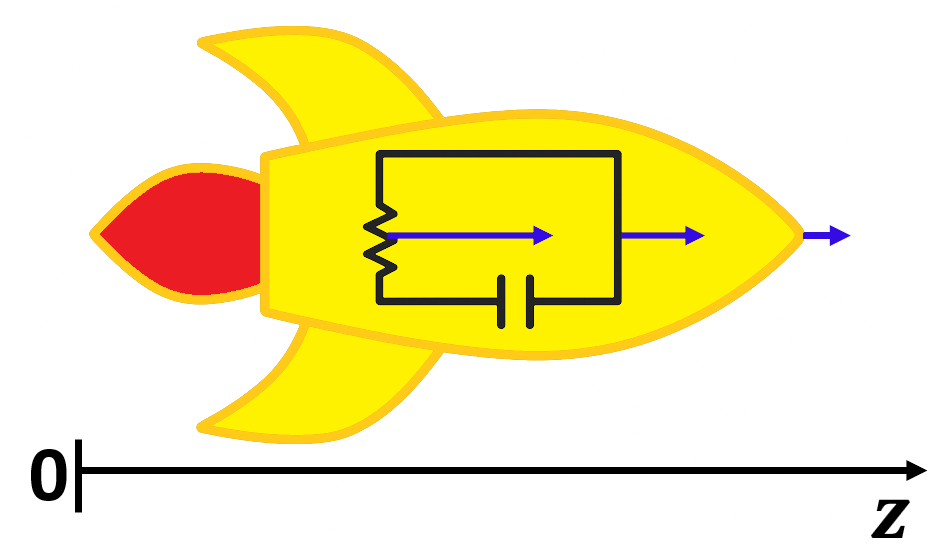}
\caption{A rocket containing a circuit undergoes rigid motion in Minkowski space ($ds^2={-}d\mathcal{T}^2{+}dx^2{+}dy^2{+}d\mathcal{Z}^2$), with velocity $u^\mu$ proportional to the generator of boosts, $\mathcal{Z}\partial_\mathcal{T}{+}\mathcal{T}\partial_\mathcal{Z}$. In Rindler coordinates ($\mathcal{T}{=}z\sinh(gt)$, $\mathcal{Z}{=}z\cosh(gt)$), we have $u^\mu\partial_\mu =(gz)^{-1}\partial_t$, and the proper acceleration  (blue arrows) of a piece of equipment is $a\,{=}\,1/z$, which increases towards the rear of the rocket, and diverges at Rindler's horizon ($z=0$).}
    \label{fig:rocket}
\end{figure}

\noindent\textit{\bf Application 1: Inertia of the electron -} We consider a piece of metal with uniform transport properties  ($\sigma_1{=}\text{const}$), which is externally kept in an isothermal state ($\K T=\text{const}$, $\nabla_\mu W^\mu=\K \mathcal{E}_\mu J^\mu +\text{``external heat''}$). Assuming invariance under $xy$-translations, we ask: How does the charge tend to arrange? To answer this question, we assume that $F=F_{tz}(t,z) dt\wedge dz$, which trivially solves the homogeneous Maxwell's equations ($dF=0$), and we write $\nabla_\mu F^{\nu \mu}=J^\nu$ explicitly, which gives
\begin{equation}\label{hyperbolz}
\begin{split}
\partial_z \mathcal{E}^z={}& \rho_\text{lte}+\sigma_1 (gz)^{-1}\partial_t\mu \, , \\
\partial_t \mathcal{E}^z={}& \sigma_1 [\partial_z(gz\mu)-gz\mathcal{E}^z] \, . \\
\end{split}
\end{equation}
To close the system, we need an equation relating $\rho_\text{lte}$ and $\mu$. We posit $\mu \,{=}\,{-}\mu_0\,{+}\, b^2\rho_\text{lte}$ (with $\mu_0$ and $b$ positive constants, estimated in the Supplementary Material), which is valid when $\rho_\text{lte}$ is small and the positive charges are equally spaced.  
Moreover, instead of solving \eqref{hyperbolz} exactly, we recall that $\sigma_1\K^\mu \nabla_\mu \mu =\mathcal{O}(\tau^2)$, so the term $\partial_t\mu$ can be neglected, and we obtain a single parabolic equation: 
\begin{equation}\label{ezparab}
\dfrac{\partial_t \mathcal{E}^z}{g\sigma_1}=-\mu_0 +b^2\, \partial_z (z\partial_z \mathcal{E}^z)-z\mathcal{E}^z \, .
\end{equation}
We solve this equation numerically, assuming local charge neutrality at $t=0$, and requiring that no charge leaves the chunk of metal, so $\mathcal{E}^z(\text{outside})=0$. The resulting evolution is shown in figure \ref{fig:parabolic}. Over time, the electrons \textit{fall behind}, due to their inertia, and they accumulate on the rear. This stops when a sufficiently strong electric field is formed, which counteracts the inertia. This is an instance of the Stewart-Tolman effect \cite{TolmanStewart1916} in relativity. 
\newpage

\begin{figure}
    \centering
\includegraphics[width=0.99\linewidth]{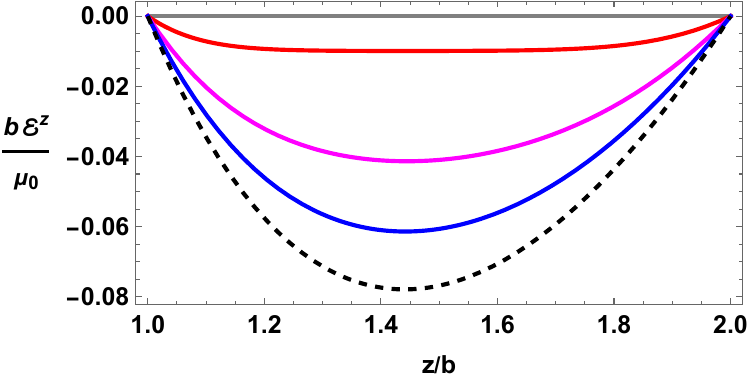}
\includegraphics[width=0.99\linewidth]{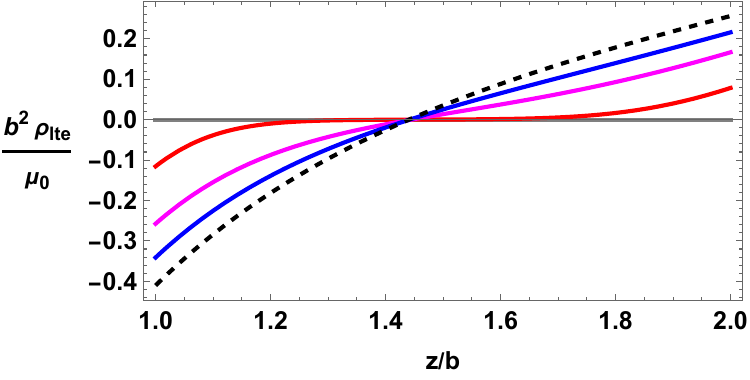}
    \caption{Electric field (upper panel) and charge density (lower panel) of a solution of \eqref{ezparab} with $\mathcal{E}^z(t{=}0)=\mathcal{E}^z(\text{boundary})=0$, for a piece of metal extending from $z=b$ to $z=2b$. The curves represent snapshots at $g\sigma_1 b t=0$ (gray), $0.01$ (red), $0.05$ (magenta), $0.1$ (blue), $\infty$ (dashed).}
    \label{fig:parabolic}
\end{figure}

\noindent\textit{\bf Application 2: Time dilation and Joule heating -} From Fig.~\ref{fig:redshifcurrent}, we know that a stationary current $I$ in an accelerated circuit scales like $1/z$, by time dilation. But since the Joule power output grows like $I^2\sim 1/z^2$, we expect wires to get warmer towards the rear of the rocket, by an amount that exceeds the redshift effect (which only predicts $T\sim 1/z$). Let us quantify this effect.
\begin{figure}[t!]
    \centering
\includegraphics[width=0.9\linewidth]{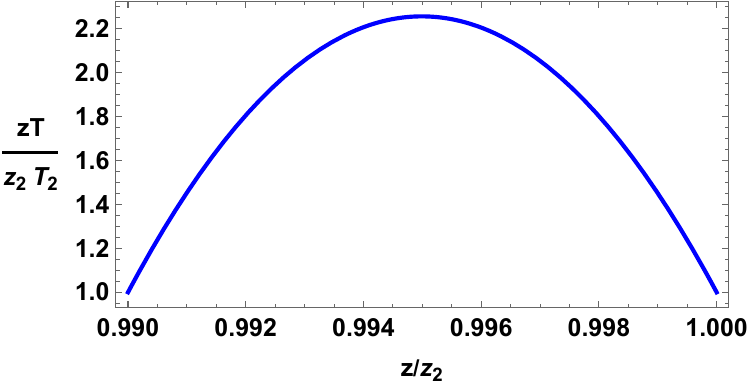}
\includegraphics[width=0.9\linewidth]{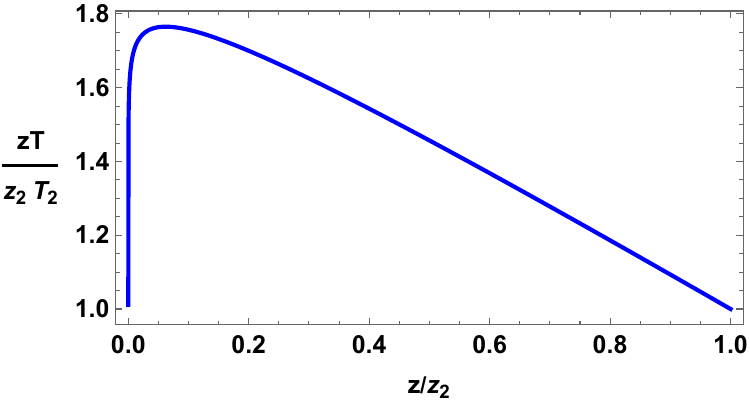}
\caption{Solutions of \eqref{dzdzT} with boundary conditions $z_1T(z_1){=}z_2 T(z_2)$. In the upper panel, we set $z_1/z_2{=}0.99$ and $\mathcal{J}^2/[\sigma_1\kappa_3T(z_2)]{=}10^5$. In the lower panel, we set $z_1/z_2{=}10^{-7}$, and $\mathcal{J}^2/[\sigma_1\kappa_3T(z_2)]=1$. }
    \label{fig:temeprature}
\end{figure}

Consider an isolated, stationary, non-equilibrium system, where two pieces of equipment located respectively at $z_1$ and $z_2$, and having the same redshifted temperature ($z_1T_1=z_2T_2$), are connected by a straight cable crossed by a uniform redshifted current density $zJ^z {=}\text{const} {\equiv} \mathcal{J}$. Let us estimate the temperature along the cable. Assuming $\sigma_3{=}0$ and $\kappa_3{=}\text{const}$, equations \eqref{entroycurrent} and \eqref{secondlaw} become
\vspace{-0.1cm}
\begin{equation}\label{sz}
zs^z{=}-\dfrac{\kappa_3}{T} \partial_z (zT) \, ,\quad
\dfrac{T}{z}\partial_z(zs^z)=\dfrac{\mathcal{J}^2}{\sigma_1 z^2}\,{+}\,\dfrac{\kappa_3}{Tz^2} [\partial_z(zT)]^2 .
\end{equation}
Combining them together, we arrive at the equation
\vspace{-0.1cm}
\begin{equation}\label{dzdzT}
\partial_z\left[z\partial_z(zT)\right]=-\dfrac{\mathcal{J}^2}{\sigma_1\kappa_3} \, ,
\end{equation}
to be solved with boundary conditions $T(z_1)\,{=}\,T_1$ and $T(z_2)\,{=}\,T_2$. In figure \ref{fig:temeprature}, we provide two solutions: One where the acceleration is negligible ($z_1\, {\rightarrow}\, \infty$), and one where it matters ($z_1\,{\rightarrow}\, 0$). In the former, the cable’s hottest spot lies at its midpoint, which is the farthest point from the thermostatic boundaries, where the cable gives off its heat. In the latter, the peak in redshifted temperature has moved notably to the left, as predicted.

\newpage
\noindent\textit{\bf Application 3: Redshifted magnetic diffusion -} Let us see how the magnetic diffusion equation \eqref{magneticdiffusion} changes on our rocket.
We consider a transversal electromagnetic wave whose four-potential is $A=A_x(t,z)dx$. Then, the Faraday tensor is $F=\partial_t A_x dt\wedge dx+\partial_z A_x dz \wedge dx$, and the only relevant Maxwell equation is $J^x=\nabla_\mu F^{x \mu}$. Since nothing depends on $x$, we have $J^x{=}\sigma_1\mathcal{E}^x$, and thus 
\begin{equation}\label{ptAeq1s1}
\partial_t A_x=\dfrac{1}{\sigma_1} \partial_z \left[ gz\partial_z A_x\right] -\dfrac{1}{\sigma_1 g z} \partial^2_t A_x \, ,
\end{equation}
where we may neglect the term with the second time derivative.
Applying $\partial_z$ on both sides, and recalling that $\mathcal{B}^y=-\varepsilon^{yzxt}F_{zx}u_t=\partial_z A_x$ \cite{special_in_gen}, we arrive at
\begin{equation}\label{grasch}
\partial_t \mathcal{B}^y=\dfrac{1}{\sigma_1} \partial_z^2 \left( gz\mathcal{B}^y\right) \, ,
\end{equation}
which is the magnetic diffusion equation we were looking for. Its stationary solutions take the form $\mathcal{B}^y= c_1+c_2/z $, so both $\mathcal{B}^y=\text{const}$ (a uniform magnetic field) and $gz\mathcal{B}^y=\text{const}$ (a uniform redshifted magnetic field) are admissible steady states. The question is then which of these represents the \textit{true} equilibrium, i.e. the one that does not rely on continuous dissipation to be sustained. The answer becomes evident once we observe that the solution of \eqref{ptAeq1s1} consistent with $\mathcal{B}^y= c_1+c_2/z $ is $A_x = c_1 \left(z+gt/{\sigma_1}\right) + c_2 \ln z $, implying that the entropy production is $\nabla_\mu s^\mu \propto (J_x)^2 \propto (\mathcal{E}_x)^2 \propto (\partial_t A_x)^2 \propto (c_1)^2$. This shows that keeping $\mathcal{B}^y$ uniform necessarily entails a continuous dissipation of heat, while the true diffusive equilibrium is instead the configuration in which the \textit{redshifted} magnetic field $gz\mathcal{B}^y$ is uniform.

\noindent\textit{\bf Application 4: Charged relativistic stars -} Let us step out of our rocket, and land on a problem of astrophysical interest. In most descriptions of charged compact objects, one assumes either $\rho_\text{lte}=\text{const} \times \varepsilon_\text{lte}$ (for bulk charge), or $\rho_\text{lte}\,{\sim}\, e^{-(r-R_\star)^2/b^2}$ (for surface charge) \cite{Ray:2003gt,Ray:2004yr,Negreiros:2009fd,Carvalho:2018kbt,FRocha2020,Arbanil2022}. These choices are made for convenience, since modeling the actual charge stratification would require a fully relativistic Thomas-Fermi equation for the electron cloud. Our theory produces such an equation naturally.

We work in a static, spherically symmetric background spacetime ($ds^2\,{=}\,-e^{2\Phi(r)}dt^2\,{+}\,e^{2\Lambda(r)}dr^2\,{+}\,r^2 d^2\Omega$), and we assume that the Faraday tensor is also spherically symmetric (but not necessarily static), so that $F\,{=}\,F_{tr}(t,r)\,dt\,{\wedge} \,dr$, which automatically fulfills $dF=0$. Then, the non-trivial Maxwell equations are $\nabla_\mu F^{t\mu}\,{=}\,J^t$ and $\nabla_\mu F^{r\mu}\,{=}\,J^r$, which explicitly read
\vspace{-0.1cm}
\begin{equation}\label{maxsphere}
\begin{split}
&\dfrac{e^{-\Lambda}}{r^2} \partial_r \left(r^2 e^{\Lambda}\mathcal{E}^r\right)= \rho_\text{lte} +\sigma_1 e^{-\Phi}\partial_t \mu+\sigma_3 e^{-\Phi}\partial_t T \, ,\\
&e^{-\Phi}\partial_t \mathcal{E}_r =\sigma_1\left[e^{-\Phi}\partial_r (e^\Phi \mu){-}\mathcal{E}_r\right]\!{+}\sigma_3e^{-\Phi}\partial_r \left(e^\Phi T\right)\, .\\
\end{split}
\end{equation}
In a time-independent scenario, one can combine the two lines of \eqref{maxsphere} into a single equation:
\vspace{-0.1cm}
\begin{equation}\label{debyeHuck}
\dfrac{e^{-\Lambda}}{r^2}\partial_r \left\{r^2 e^{-\Lambda-\Phi}\left[\partial_r \left( e^\Phi \mu\right)+\dfrac{\sigma_3}{\sigma_1}\partial_r \left( e^\Phi T\right)\right]\right\}=\rho_\text{lte} \, ,
\end{equation}
which is of Thomas-Fermi type \cite{Thomas_1927,Fermi1927}.
Gradients of $e^\Phi T$ are often present since stars cool down or accrete. The Seebeck coefficient $-\sigma_3/\sigma_1$ (or ``thermopower'') has been computed from microphysics in several references \cite{Geppert1991,Yakovlev1984,Schmitt:2017efp,Harutyunyan:2024hsd,Bajec:2025dqm}, but the most common expression in neutron star models is \cite{Gakis:2024obw} ($n_e =$ ``electron number density'')
\begin{equation}
\dfrac{\sigma_3}{\sigma_1} =-\left(\dfrac{\pi^4}{3n_e e^3} \right)^{1/3} T \, .
\end{equation}
In the cold limit, we can neglect the Seebeck correction, and \eqref{debyeHuck} becomes a differential equation for $\mu$ alone,
to be solved with boundary conditions $[\partial_r\mu]_{r=0}\,{=}\,0$ and $4\pi R_\star^2 \left[e^{-\Lambda-\Phi}\partial_r \left( e^\Phi \mu\right)\right]_{r=R_\star} =Q=\text{``total charge''}$. 

In figure \ref{fig:compactness}, we test the assumption $\rho_\text{lte}=\text{const} \times \varepsilon_\text{lte}$ against our theory. Specifically, we consider a low temperature, constant density star ($\varepsilon_\text{lte}{=}\text{const}$) with compactness $C{=}0.33$, corresponding to the metric \cite[\S 11.6]{Weinberg_book_1972}
\vspace{-0.1cm}
\begin{equation}
\begin{split}
 e^{-\Lambda(r)}={}& \sqrt{1{-}2C \tfrac{r^2}{R_\star^2}\,} \, , \\ 
 e^{\Phi(r)}={}& \tfrac{3}{2}e^{-\Lambda(R_\star)}-\tfrac{1}{2}e^{-\Lambda(r)}\, , \\
\end{split}
\end{equation}
and we ``prepare'' the electron cloud so that $\rho_\text{lte}=\text{const}$. Then, we let the system evolve according to \eqref{maxsphere}, with $\sigma_1=\text{const}$ and $\mu=-\mu_0+b^2\rho_\text{lte}$ ($\mu_0=\text{const}$, $b=0.1\, R_\star$).

As expected, the electrons spontaneously migrate to a solution of the Thomas-Fermi equation \eqref{debyeHuck}. The final configuration is a balance between two competing effects: the familiar tendency of charge to accumulate near the surface, within a layer of thickness $b$, due to electrostatic repulsion, and the opposing tendency of electrons to concentrate at the center under the influence of gravity.

\newpage
We emphasize that, since $\mu_0$ and $b$ are functions of the ion number density (see Supplementary Material), taking them to be constant is equivalent to assuming a uniform positive-ion density throughout the star, which is consistent with the constant-density idealization. In a more realistic model, the radial profiles $\mu_0(r)$ and $b(r)$ should instead be determined by the ion density obtained from solving the Tolman--Oppenheimer--Volkoff equations for the equation of state of interest.




\begin{figure}[h!]
    \centering
\includegraphics[width=0.97\linewidth]{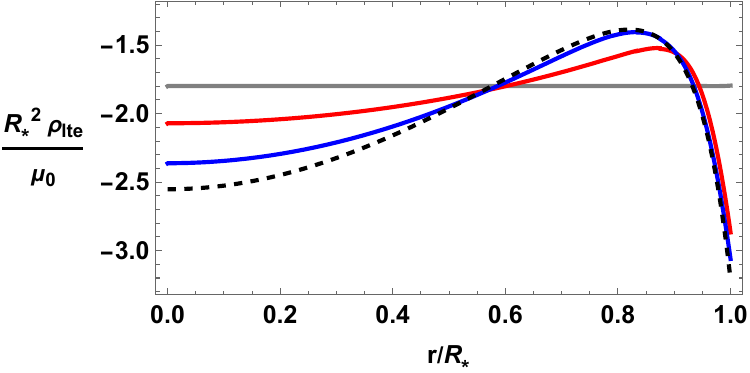}
\caption{Density profile of a solution of \eqref{maxsphere} with $\rho(t{=}0)=\text{const}$ and $\mathcal{E}_r(r{=}R_\star)=\text{const}$ (no charge escapes the star), both adjusted so that $Q/(R_\star\mu_0)=-10$. The curves represent snapshots at $\sigma_1 t=0$ (gray), $1$ (red), $3$ (blue), and $\infty$ (dashed). The dashed line also solves \eqref{debyeHuck}.}
    \label{fig:compactness}
\end{figure}

\newpage
\noindent\textit{\bf Conclusions -} We have presented a robust framework for modeling electrothermal effects in strong gravitational fields, providing a unified description of diverse phenomena, ranging from the magnetothermal evolution of neutron stars to the electric stratification of quark stars. Our framework can be straightforwardly extended to incorporate bound charges, neutrino luminosities, and anisotropies arising from strong magnetic fields or the solid lattice. The modified equations with a strong magnetic field are provided in the Supplementary Material.

Beyond its astrophysical relevance, the theory is of fundamental interest even in the context of inertial systems in flat spacetime. In fact, classical electrodynamics is the archetype of a causal, special-relativistic field theory with hyperbolic equations of motion. Coupling it to an acausal parabolic matter sector (as is done in textbooks) breaks this link with relativity. Our framework restores full consistency, while still reproducing the known phenomenology of conductors in the appropriate limits.

Owing to the assumption of Born rigidity, this theory applies most naturally to solids. It can nevertheless be viewed as a dynamical sector of resistive magnetohydrodynamics, describing thermoelectric processes in which the fluid motion is approximated by a Killing flow.

\section*{Acknowledgements}

This work is supported by a MERAC Foundation prize grant,  an Isaac Newton Trust Grant, and funding from the Cambridge Centre for Theoretical Cosmology.

\bibliography{Biblio}

\begin{thebibliography}{69}%
\makeatletter
\providecommand \@ifxundefined [1]{%
 \@ifx{#1\undefined}
}%
\providecommand \@ifnum [1]{%
 \ifnum #1\expandafter \@firstoftwo
 \else \expandafter \@secondoftwo
 \fi
}%
\providecommand \@ifx [1]{%
 \ifx #1\expandafter \@firstoftwo
 \else \expandafter \@secondoftwo
 \fi
}%
\providecommand \natexlab [1]{#1}%
\providecommand \enquote  [1]{``#1''}%
\providecommand \bibnamefont  [1]{#1}%
\providecommand \bibfnamefont [1]{#1}%
\providecommand \citenamefont [1]{#1}%
\providecommand \href@noop [0]{\@secondoftwo}%
\providecommand \href [0]{\begingroup \@sanitize@url \@href}%
\providecommand \@href[1]{\@@startlink{#1}\@@href}%
\providecommand \@@href[1]{\endgroup#1\@@endlink}%
\providecommand \@sanitize@url [0]{\catcode `\\12\catcode `\$12\catcode `\&12\catcode `\#12\catcode `\^12\catcode `\_12\catcode `\%12\relax}%
\providecommand \@@startlink[1]{}%
\providecommand \@@endlink[0]{}%
\providecommand \url  [0]{\begingroup\@sanitize@url \@url }%
\providecommand \@url [1]{\endgroup\@href {#1}{\urlprefix }}%
\providecommand \urlprefix  [0]{URL }%
\providecommand \Eprint [0]{\href }%
\providecommand \doibase [0]{http://dx.doi.org/}%
\providecommand \selectlanguage [0]{\@gobble}%
\providecommand \bibinfo  [0]{\@secondoftwo}%
\providecommand \bibfield  [0]{\@secondoftwo}%
\providecommand \translation [1]{[#1]}%
\providecommand \BibitemOpen [0]{}%
\providecommand \bibitemStop [0]{}%
\providecommand \bibitemNoStop [0]{.\EOS\space}%
\providecommand \EOS [0]{\spacefactor3000\relax}%
\providecommand \BibitemShut  [1]{\csname bibitem#1\endcsname}%
\let\auto@bib@innerbib\@empty
\bibitem [{\citenamefont {Tolman}\ and\ \citenamefont {Stewart}(1916)}]{TolmanStewart1916}%
  \BibitemOpen
  \bibfield  {author} {\bibinfo {author} {\bibfnamefont {R.~C.}\ \bibnamefont {Tolman}}\ and\ \bibinfo {author} {\bibfnamefont {T.~D.}\ \bibnamefont {Stewart}},\ }\href {\doibase 10.1103/PhysRev.8.97} {\bibfield  {journal} {\bibinfo  {journal} {Phys. Rev.}\ }\textbf {\bibinfo {volume} {8}},\ \bibinfo {pages} {97} (\bibinfo {year} {1916})}\BibitemShut {NoStop}%
\bibitem [{\citenamefont {DeWitt}(1966)}]{DeWittDrag1966}%
  \BibitemOpen
  \bibfield  {author} {\bibinfo {author} {\bibfnamefont {B.~S.}\ \bibnamefont {DeWitt}},\ }\href {\doibase 10.1103/PhysRevLett.16.1092} {\bibfield  {journal} {\bibinfo  {journal} {Phys. Rev. Lett.}\ }\textbf {\bibinfo {volume} {16}},\ \bibinfo {pages} {1092} (\bibinfo {year} {1966})}\BibitemShut {NoStop}%
\bibitem [{\citenamefont {Anandan}(1984)}]{Anandan1984}%
  \BibitemOpen
  \bibfield  {author} {\bibinfo {author} {\bibfnamefont {J.}~\bibnamefont {Anandan}},\ }\href {\doibase https://doi.org/10.1016/0375-9601(84)90997-6} {\bibfield  {journal} {\bibinfo  {journal} {Physics Letters A}\ }\textbf {\bibinfo {volume} {105}},\ \bibinfo {pages} {280} (\bibinfo {year} {1984})}\BibitemShut {NoStop}%
\bibitem [{\citenamefont {Ahmedov}(1998)}]{Ahmedov:1998mif}%
  \BibitemOpen
  \bibfield  {author} {\bibinfo {author} {\bibfnamefont {B.~J.}\ \bibnamefont {Ahmedov}},\ }\href@noop {} {\bibfield  {journal} {\bibinfo  {journal} {Grav. Cosmol.}\ }\textbf {\bibinfo {volume} {4}},\ \bibinfo {pages} {139} (\bibinfo {year} {1998})},\ \Eprint {http://arxiv.org/abs/gr-qc/0701046} {arXiv:gr-qc/0701046} \BibitemShut {NoStop}%
\bibitem [{\citenamefont {Ahmedov}(1999)}]{Ahmedov:1999bqr}%
  \BibitemOpen
  \bibfield  {author} {\bibinfo {author} {\bibfnamefont {B.~J.}\ \bibnamefont {Ahmedov}},\ }\href {\doibase 10.1023/A:1026692711377} {\bibfield  {journal} {\bibinfo  {journal} {Gen. Rel. Grav.}\ }\textbf {\bibinfo {volume} {31}},\ \bibinfo {pages} {357} (\bibinfo {year} {1999})},\ \Eprint {http://arxiv.org/abs/gr-qc/0701072} {arXiv:gr-qc/0701072} \BibitemShut {NoStop}%
\bibitem [{\citenamefont {Ahmedov}\ and\ \citenamefont {Ermamatov}(2002)}]{Ahmedov:2002bx}%
  \BibitemOpen
  \bibfield  {author} {\bibinfo {author} {\bibfnamefont {B.~J.}\ \bibnamefont {Ahmedov}}\ and\ \bibinfo {author} {\bibfnamefont {M.~J.}\ \bibnamefont {Ermamatov}},\ }\href@noop {} {\bibfield  {journal} {\bibinfo  {journal} {Phys. Lett.}\ }\textbf {\bibinfo {volume} {15}},\ \bibinfo {pages} {137} (\bibinfo {year} {2002})},\ \Eprint {http://arxiv.org/abs/gr-qc/0608060} {arXiv:gr-qc/0608060} \BibitemShut {NoStop}%
\bibitem [{\citenamefont {Ummarino}\ and\ \citenamefont {Gallerati}(2017)}]{Ummarino:2017bvz}%
  \BibitemOpen
  \bibfield  {author} {\bibinfo {author} {\bibfnamefont {G.~A.}\ \bibnamefont {Ummarino}}\ and\ \bibinfo {author} {\bibfnamefont {A.}~\bibnamefont {Gallerati}},\ }\href {\doibase 10.1140/epjc/s10052-017-5116-y} {\bibfield  {journal} {\bibinfo  {journal} {Eur. Phys. J. C}\ }\textbf {\bibinfo {volume} {77}},\ \bibinfo {pages} {549} (\bibinfo {year} {2017})},\ \Eprint {http://arxiv.org/abs/1710.01267} {arXiv:1710.01267 [gr-qc]} \BibitemShut {NoStop}%
\bibitem [{\citenamefont {{Wang}}(2023)}]{WangStewartTolman2023}%
  \BibitemOpen
  \bibfield  {author} {\bibinfo {author} {\bibfnamefont {Z.-Y.}\ \bibnamefont {{Wang}}},\ }\href {\doibase 10.1007/s00339-022-06344-9} {\bibfield  {journal} {\bibinfo  {journal} {Applied Physics A: Materials Science \& Processing}\ }\textbf {\bibinfo {volume} {129}},\ \bibinfo {eid} {71} (\bibinfo {year} {2023})}\BibitemShut {NoStop}%
\bibitem [{\citenamefont {Bei}(2025)}]{BeiPhD}%
  \BibitemOpen
  \bibfield  {author} {\bibinfo {author} {\bibfnamefont {G.}~\bibnamefont {Bei}},\ }\emph {\bibinfo {title} {Diffusione anisotropa e ondulatoria del calore e effetti termomagnetici e termoelettrici oscillanti autoindotti sui conduttori rotanti.}},\ \href@noop {} {Ph.D. thesis},\ \bibinfo  {school} {Universit\'{a} di Roma Sapienza} (\bibinfo {year} {2025})\BibitemShut {NoStop}%
\bibitem [{\citenamefont {Gavassino}(2025)}]{Gavassino:2025wcy}%
  \BibitemOpen
  \bibfield  {author} {\bibinfo {author} {\bibfnamefont {L.}~\bibnamefont {Gavassino}},\ }\href@noop {} {\  (\bibinfo {year} {2025})},\ \Eprint {http://arxiv.org/abs/2508.08936} {arXiv:2508.08936 [gr-qc]} \BibitemShut {NoStop}%
\bibitem [{\citenamefont {Schiff}\ and\ \citenamefont {Barnhill}(1966)}]{ShiffBarhill1966}%
  \BibitemOpen
  \bibfield  {author} {\bibinfo {author} {\bibfnamefont {L.~I.}\ \bibnamefont {Schiff}}\ and\ \bibinfo {author} {\bibfnamefont {M.~V.}\ \bibnamefont {Barnhill}},\ }\href {\doibase 10.1103/PhysRev.151.1067} {\bibfield  {journal} {\bibinfo  {journal} {Phys. Rev.}\ }\textbf {\bibinfo {volume} {151}},\ \bibinfo {pages} {1067} (\bibinfo {year} {1966})}\BibitemShut {NoStop}%
\bibitem [{\citenamefont {{Bhatia}}\ \emph {et~al.}(1969)\citenamefont {{Bhatia}}, \citenamefont {{Bonazzola}},\ and\ \citenamefont {{Szamosi}}}]{Bhatia1969}%
  \BibitemOpen
  \bibfield  {author} {\bibinfo {author} {\bibfnamefont {M.~S.}\ \bibnamefont {{Bhatia}}}, \bibinfo {author} {\bibfnamefont {S.}~\bibnamefont {{Bonazzola}}}, \ and\ \bibinfo {author} {\bibfnamefont {G.}~\bibnamefont {{Szamosi}}},\ }\href@noop {} {\bibfield  {journal} {\bibinfo  {journal} {\aap}\ }\textbf {\bibinfo {volume} {3}},\ \bibinfo {pages} {206} (\bibinfo {year} {1969})}\BibitemShut {NoStop}%
\bibitem [{\citenamefont {{Bekenstein}}(1971)}]{Bekenstein1971}%
  \BibitemOpen
  \bibfield  {author} {\bibinfo {author} {\bibfnamefont {J.~D.}\ \bibnamefont {{Bekenstein}}},\ }\href {\doibase 10.1103/PhysRevD.4.2185} {\bibfield  {journal} {\bibinfo  {journal} {\prd}\ }\textbf {\bibinfo {volume} {4}},\ \bibinfo {pages} {2185} (\bibinfo {year} {1971})}\BibitemShut {NoStop}%
\bibitem [{\citenamefont {{Zhang}}\ \emph {et~al.}(1982)\citenamefont {{Zhang}}, \citenamefont {{Chau}},\ and\ \citenamefont {{Deng}}}]{Zhand1982}%
  \BibitemOpen
  \bibfield  {author} {\bibinfo {author} {\bibfnamefont {J.~L.}\ \bibnamefont {{Zhang}}}, \bibinfo {author} {\bibfnamefont {W.~Y.}\ \bibnamefont {{Chau}}}, \ and\ \bibinfo {author} {\bibfnamefont {T.~Y.}\ \bibnamefont {{Deng}}},\ }\href {\doibase 10.1007/BF00648990} {\bibfield  {journal} {\bibinfo  {journal} {\apss}\ }\textbf {\bibinfo {volume} {88}},\ \bibinfo {pages} {81} (\bibinfo {year} {1982})}\BibitemShut {NoStop}%
\bibitem [{\citenamefont {{Alcock}}\ \emph {et~al.}(1986)\citenamefont {{Alcock}}, \citenamefont {{Farhi}},\ and\ \citenamefont {{Olinto}}}]{Alcock1986}%
  \BibitemOpen
  \bibfield  {author} {\bibinfo {author} {\bibfnamefont {C.}~\bibnamefont {{Alcock}}}, \bibinfo {author} {\bibfnamefont {E.}~\bibnamefont {{Farhi}}}, \ and\ \bibinfo {author} {\bibfnamefont {A.}~\bibnamefont {{Olinto}}},\ }\href {\doibase 10.1086/164679} {\bibfield  {journal} {\bibinfo  {journal} {\apj}\ }\textbf {\bibinfo {volume} {310}},\ \bibinfo {pages} {261} (\bibinfo {year} {1986})}\BibitemShut {NoStop}%
\bibitem [{\citenamefont {Usov}(2004)}]{Usov:2004iz}%
  \BibitemOpen
  \bibfield  {author} {\bibinfo {author} {\bibfnamefont {V.~V.}\ \bibnamefont {Usov}},\ }\href {\doibase 10.1103/PhysRevD.70.067301} {\bibfield  {journal} {\bibinfo  {journal} {Phys. Rev. D}\ }\textbf {\bibinfo {volume} {70}},\ \bibinfo {pages} {067301} (\bibinfo {year} {2004})},\ \Eprint {http://arxiv.org/abs/astro-ph/0408217} {arXiv:astro-ph/0408217} \BibitemShut {NoStop}%
\bibitem [{\citenamefont {Usov}\ \emph {et~al.}(2005)\citenamefont {Usov}, \citenamefont {Harko},\ and\ \citenamefont {Cheng}}]{Usov:2004kj}%
  \BibitemOpen
  \bibfield  {author} {\bibinfo {author} {\bibfnamefont {V.~V.}\ \bibnamefont {Usov}}, \bibinfo {author} {\bibfnamefont {T.}~\bibnamefont {Harko}}, \ and\ \bibinfo {author} {\bibfnamefont {K.~S.}\ \bibnamefont {Cheng}},\ }\href {\doibase 10.1086/427074} {\bibfield  {journal} {\bibinfo  {journal} {Astrophys. J.}\ }\textbf {\bibinfo {volume} {620}},\ \bibinfo {pages} {915} (\bibinfo {year} {2005})},\ \Eprint {http://arxiv.org/abs/astro-ph/0410682} {arXiv:astro-ph/0410682} \BibitemShut {NoStop}%
\bibitem [{\citenamefont {Ray}\ \emph {et~al.}(2003)\citenamefont {Ray}, \citenamefont {Espindola}, \citenamefont {Malheiro}, \citenamefont {Lemos},\ and\ \citenamefont {Zanchin}}]{Ray:2003gt}%
  \BibitemOpen
  \bibfield  {author} {\bibinfo {author} {\bibfnamefont {S.}~\bibnamefont {Ray}}, \bibinfo {author} {\bibfnamefont {A.~L.}\ \bibnamefont {Espindola}}, \bibinfo {author} {\bibfnamefont {M.}~\bibnamefont {Malheiro}}, \bibinfo {author} {\bibfnamefont {J.~P.~S.}\ \bibnamefont {Lemos}}, \ and\ \bibinfo {author} {\bibfnamefont {V.~T.}\ \bibnamefont {Zanchin}},\ }\href {\doibase 10.1103/PhysRevD.68.084004} {\bibfield  {journal} {\bibinfo  {journal} {Phys. Rev. D}\ }\textbf {\bibinfo {volume} {68}},\ \bibinfo {pages} {084004} (\bibinfo {year} {2003})},\ \Eprint {http://arxiv.org/abs/astro-ph/0307262} {arXiv:astro-ph/0307262} \BibitemShut {NoStop}%
\bibitem [{\citenamefont {Ray}\ \emph {et~al.}(2004)\citenamefont {Ray}, \citenamefont {Malheiro}, \citenamefont {Lemos},\ and\ \citenamefont {Zanchin}}]{Ray:2004yr}%
  \BibitemOpen
  \bibfield  {author} {\bibinfo {author} {\bibfnamefont {S.}~\bibnamefont {Ray}}, \bibinfo {author} {\bibfnamefont {M.}~\bibnamefont {Malheiro}}, \bibinfo {author} {\bibfnamefont {J.~P.~S.}\ \bibnamefont {Lemos}}, \ and\ \bibinfo {author} {\bibfnamefont {V.~T.}\ \bibnamefont {Zanchin}},\ }\href {\doibase 10.1590/S0103-97332004000200038} {\bibfield  {journal} {\bibinfo  {journal} {Braz. J. Phys.}\ }\textbf {\bibinfo {volume} {34}},\ \bibinfo {pages} {310} (\bibinfo {year} {2004})},\ \Eprint {http://arxiv.org/abs/nucl-th/0403056} {arXiv:nucl-th/0403056} \BibitemShut {NoStop}%
\bibitem [{\citenamefont {Negreiros}\ \emph {et~al.}(2009)\citenamefont {Negreiros}, \citenamefont {Weber}, \citenamefont {Malheiro},\ and\ \citenamefont {Usov}}]{Negreiros:2009fd}%
  \BibitemOpen
  \bibfield  {author} {\bibinfo {author} {\bibfnamefont {R.~P.}\ \bibnamefont {Negreiros}}, \bibinfo {author} {\bibfnamefont {F.}~\bibnamefont {Weber}}, \bibinfo {author} {\bibfnamefont {M.}~\bibnamefont {Malheiro}}, \ and\ \bibinfo {author} {\bibfnamefont {V.}~\bibnamefont {Usov}},\ }\href {\doibase 10.1103/PhysRevD.80.083006} {\bibfield  {journal} {\bibinfo  {journal} {Phys. Rev. D}\ }\textbf {\bibinfo {volume} {80}},\ \bibinfo {pages} {083006} (\bibinfo {year} {2009})},\ \Eprint {http://arxiv.org/abs/0907.5537} {arXiv:0907.5537 [astro-ph.SR]} \BibitemShut {NoStop}%
\bibitem [{\citenamefont {{Yousaf}}\ and\ \citenamefont {{Bhatti}}(2016)}]{Yousaf2016}%
  \BibitemOpen
  \bibfield  {author} {\bibinfo {author} {\bibfnamefont {Z.}~\bibnamefont {{Yousaf}}}\ and\ \bibinfo {author} {\bibfnamefont {M.~Z. u.~H.}\ \bibnamefont {{Bhatti}}},\ }\href {\doibase 10.1093/mnras/stw423} {\bibfield  {journal} {\bibinfo  {journal} {\mnras}\ }\textbf {\bibinfo {volume} {458}},\ \bibinfo {pages} {1785} (\bibinfo {year} {2016})},\ \Eprint {http://arxiv.org/abs/1612.02325} {arXiv:1612.02325 [physics.gen-ph]} \BibitemShut {NoStop}%
\bibitem [{\citenamefont {Carvalho}\ \emph {et~al.}(2018)\citenamefont {Carvalho}, \citenamefont {Arba{\~n}il}, \citenamefont {Marinho},\ and\ \citenamefont {Malheiro}}]{Carvalho:2018kbt}%
  \BibitemOpen
  \bibfield  {author} {\bibinfo {author} {\bibfnamefont {G.~A.}\ \bibnamefont {Carvalho}}, \bibinfo {author} {\bibfnamefont {J.~D.~V.}\ \bibnamefont {Arba{\~n}il}}, \bibinfo {author} {\bibfnamefont {R.~M.}\ \bibnamefont {Marinho}}, \ and\ \bibinfo {author} {\bibfnamefont {M.}~\bibnamefont {Malheiro}},\ }\href {\doibase 10.1140/epjc/s10052-018-5901-2} {\bibfield  {journal} {\bibinfo  {journal} {Eur. Phys. J. C}\ }\textbf {\bibinfo {volume} {78}},\ \bibinfo {pages} {411} (\bibinfo {year} {2018})},\ \Eprint {http://arxiv.org/abs/1805.07257} {arXiv:1805.07257 [astro-ph.SR]} \BibitemShut {NoStop}%
\bibitem [{\citenamefont {Rocha}\ \emph {et~al.}(2020)\citenamefont {Rocha}, \citenamefont {Carvalho}, \citenamefont {Deb},\ and\ \citenamefont {Malheiro}}]{FRocha2020}%
  \BibitemOpen
  \bibfield  {author} {\bibinfo {author} {\bibfnamefont {F.}~\bibnamefont {Rocha}}, \bibinfo {author} {\bibfnamefont {G.~A.}\ \bibnamefont {Carvalho}}, \bibinfo {author} {\bibfnamefont {D.}~\bibnamefont {Deb}}, \ and\ \bibinfo {author} {\bibfnamefont {M.}~\bibnamefont {Malheiro}},\ }\href {\doibase 10.1103/PhysRevD.101.104008} {\bibfield  {journal} {\bibinfo  {journal} {Phys. Rev. D}\ }\textbf {\bibinfo {volume} {101}},\ \bibinfo {pages} {104008} (\bibinfo {year} {2020})}\BibitemShut {NoStop}%
\bibitem [{\citenamefont {Bhatti}\ and\ \citenamefont {Tariq}(2020)}]{Bhatti2020}%
  \BibitemOpen
  \bibfield  {author} {\bibinfo {author} {\bibfnamefont {M.~Z.}\ \bibnamefont {Bhatti}}\ and\ \bibinfo {author} {\bibfnamefont {Z.}~\bibnamefont {Tariq}},\ }\href {\doibase 10.1016/j.dark.2020.100600} {\bibfield  {journal} {\bibinfo  {journal} {Phys. Dark Univ.}\ }\textbf {\bibinfo {volume} {29}},\ \bibinfo {pages} {100600} (\bibinfo {year} {2020})},\ \Eprint {http://arxiv.org/abs/2001.06327} {arXiv:2001.06327 [gr-qc]} \BibitemShut {NoStop}%
\bibitem [{\citenamefont {{Arba{\~n}il}}\ \emph {et~al.}(2022)\citenamefont {{Arba{\~n}il}}, \citenamefont {{Carvalho}}, \citenamefont {{Marinho}},\ and\ \citenamefont {{Malheiro}}}]{Arbanil2022}%
  \BibitemOpen
  \bibfield  {author} {\bibinfo {author} {\bibfnamefont {J.~D.~V.}\ \bibnamefont {{Arba{\~n}il}}}, \bibinfo {author} {\bibfnamefont {G.~A.}\ \bibnamefont {{Carvalho}}}, \bibinfo {author} {\bibfnamefont {R.~M.}\ \bibnamefont {{Marinho}}}, \ and\ \bibinfo {author} {\bibfnamefont {M.}~\bibnamefont {{Malheiro}}},\ }in\ \href {\doibase 10.1142/9789811258251_0312} {\emph {\bibinfo {booktitle} {Fifteenth Marcel Grossmann Meeting on General Relativity}}},\ \bibinfo {editor} {edited by\ \bibinfo {editor} {\bibfnamefont {E.~S.}\ \bibnamefont {{Battistelli}}}, \bibinfo {editor} {\bibfnamefont {R.~T.}\ \bibnamefont {{Jantzen}}}, \ and\ \bibinfo {editor} {\bibfnamefont {R.}~\bibnamefont {{Ruffini}}}}\ (\bibinfo {year} {2022})\ pp.\ \bibinfo {pages} {2081--2085}\BibitemShut {NoStop}%
\bibitem [{\citenamefont {Sharif}\ and\ \citenamefont {Naseer}(2023)}]{Sharif:2023vpb}%
  \BibitemOpen
  \bibfield  {author} {\bibinfo {author} {\bibfnamefont {M.}~\bibnamefont {Sharif}}\ and\ \bibinfo {author} {\bibfnamefont {T.}~\bibnamefont {Naseer}},\ }\href {\doibase 10.1002/prop.202200147} {\bibfield  {journal} {\bibinfo  {journal} {Fortsch. Phys.}\ }\textbf {\bibinfo {volume} {71}},\ \bibinfo {pages} {2200147} (\bibinfo {year} {2023})},\ \Eprint {http://arxiv.org/abs/2303.04472} {arXiv:2303.04472 [gr-qc]} \BibitemShut {NoStop}%
\bibitem [{\citenamefont {Masa}\ \emph {et~al.}(2023)\citenamefont {Masa}, \citenamefont {Lemos},\ and\ \citenamefont {Zanchin}}]{Masa2023}%
  \BibitemOpen
  \bibfield  {author} {\bibinfo {author} {\bibfnamefont {A.~D.~D.}\ \bibnamefont {Masa}}, \bibinfo {author} {\bibfnamefont {J.~P.~S.}\ \bibnamefont {Lemos}}, \ and\ \bibinfo {author} {\bibfnamefont {V.~T.}\ \bibnamefont {Zanchin}},\ }\href {\doibase 10.1103/PhysRevD.107.064053} {\bibfield  {journal} {\bibinfo  {journal} {Phys. Rev. D}\ }\textbf {\bibinfo {volume} {107}},\ \bibinfo {pages} {064053} (\bibinfo {year} {2023})}\BibitemShut {NoStop}%
\bibitem [{\citenamefont {Jain}\ \emph {et~al.}(1987)\citenamefont {Jain}, \citenamefont {Lukens},\ and\ \citenamefont {Tsai}}]{Jain1987}%
  \BibitemOpen
  \bibfield  {author} {\bibinfo {author} {\bibfnamefont {A.~K.}\ \bibnamefont {Jain}}, \bibinfo {author} {\bibfnamefont {J.~E.}\ \bibnamefont {Lukens}}, \ and\ \bibinfo {author} {\bibfnamefont {J.~S.}\ \bibnamefont {Tsai}},\ }\href {\doibase 10.1103/PhysRevLett.58.1165} {\bibfield  {journal} {\bibinfo  {journal} {Phys. Rev. Lett.}\ }\textbf {\bibinfo {volume} {58}},\ \bibinfo {pages} {1165} (\bibinfo {year} {1987})}\BibitemShut {NoStop}%
\bibitem [{\citenamefont {Landau}\ and\ \citenamefont {Lifshitz}(1984)}]{landau8}%
  \BibitemOpen
  \bibfield  {author} {\bibinfo {author} {\bibfnamefont {L.}~\bibnamefont {Landau}}\ and\ \bibinfo {author} {\bibfnamefont {E.}~\bibnamefont {Lifshitz}},\ }\href@noop {} {\emph {\bibinfo {title} {Electrodynamics of continuous media, 2nd Edition}}},\ \bibinfo {number} {v. 8}\ (\bibinfo  {publisher} {Pergamon press},\ \bibinfo {address} {Oxford},\ \bibinfo {year} {1984})\BibitemShut {NoStop}%
\bibitem [{\citenamefont {Sommerfeld}(1964)}]{SommerfeldStatmech}%
  \BibitemOpen
  \bibfield  {author} {\bibinfo {author} {\bibfnamefont {A.}~\bibnamefont {Sommerfeld}},\ }\href@noop {} {\emph {\bibinfo {title} {Thermodynamics and Statistical Mechanics}}}\ (\bibinfo  {publisher} {Academic Press, New York},\ \bibinfo {year} {1964})\BibitemShut {NoStop}%
\bibitem [{\citenamefont {Hernandez}\ and\ \citenamefont {Kovtun}(2017)}]{Hernandez:2017mch}%
  \BibitemOpen
  \bibfield  {author} {\bibinfo {author} {\bibfnamefont {J.}~\bibnamefont {Hernandez}}\ and\ \bibinfo {author} {\bibfnamefont {P.}~\bibnamefont {Kovtun}},\ }\href {\doibase 10.1007/JHEP05(2017)001} {\bibfield  {journal} {\bibinfo  {journal} {JHEP}\ }\textbf {\bibinfo {volume} {05}},\ \bibinfo {pages} {001} (\bibinfo {year} {2017})},\ \Eprint {http://arxiv.org/abs/1703.08757} {arXiv:1703.08757 [hep-th]} \BibitemShut {NoStop}%
\bibitem [{\citenamefont {Fotakis}\ \emph {et~al.}(2020)\citenamefont {Fotakis}, \citenamefont {Greif}, \citenamefont {Greiner}, \citenamefont {Denicol},\ and\ \citenamefont {Niemi}}]{Fotakis:2019nbq}%
  \BibitemOpen
  \bibfield  {author} {\bibinfo {author} {\bibfnamefont {J.~A.}\ \bibnamefont {Fotakis}}, \bibinfo {author} {\bibfnamefont {M.}~\bibnamefont {Greif}}, \bibinfo {author} {\bibfnamefont {C.}~\bibnamefont {Greiner}}, \bibinfo {author} {\bibfnamefont {G.~S.}\ \bibnamefont {Denicol}}, \ and\ \bibinfo {author} {\bibfnamefont {H.}~\bibnamefont {Niemi}},\ }\href {\doibase 10.1103/PhysRevD.101.076007} {\bibfield  {journal} {\bibinfo  {journal} {Phys. Rev. D}\ }\textbf {\bibinfo {volume} {101}},\ \bibinfo {pages} {076007} (\bibinfo {year} {2020})},\ \Eprint {http://arxiv.org/abs/1912.09103} {arXiv:1912.09103 [hep-ph]} \BibitemShut {NoStop}%
\bibitem [{\citenamefont {Dommes}\ \emph {et~al.}(2020)\citenamefont {Dommes}, \citenamefont {Gusakov},\ and\ \citenamefont {Shternin}}]{Dommes:2020ktk}%
  \BibitemOpen
  \bibfield  {author} {\bibinfo {author} {\bibfnamefont {V.~A.}\ \bibnamefont {Dommes}}, \bibinfo {author} {\bibfnamefont {M.~E.}\ \bibnamefont {Gusakov}}, \ and\ \bibinfo {author} {\bibfnamefont {P.~S.}\ \bibnamefont {Shternin}},\ }\href {\doibase 10.1103/PhysRevD.101.103020} {\bibfield  {journal} {\bibinfo  {journal} {Phys. Rev. D}\ }\textbf {\bibinfo {volume} {101}},\ \bibinfo {pages} {103020} (\bibinfo {year} {2020})},\ \Eprint {http://arxiv.org/abs/2006.09840} {arXiv:2006.09840 [astro-ph.HE]} \BibitemShut {NoStop}%
\bibitem [{\citenamefont {Dash}\ \emph {et~al.}(2023)\citenamefont {Dash}, \citenamefont {Shokri}, \citenamefont {Rezzolla},\ and\ \citenamefont {Rischke}}]{Dash:2022xkz}%
  \BibitemOpen
  \bibfield  {author} {\bibinfo {author} {\bibfnamefont {A.}~\bibnamefont {Dash}}, \bibinfo {author} {\bibfnamefont {M.}~\bibnamefont {Shokri}}, \bibinfo {author} {\bibfnamefont {L.}~\bibnamefont {Rezzolla}}, \ and\ \bibinfo {author} {\bibfnamefont {D.~H.}\ \bibnamefont {Rischke}},\ }\href {\doibase 10.1103/PhysRevD.107.056003} {\bibfield  {journal} {\bibinfo  {journal} {Phys. Rev. D}\ }\textbf {\bibinfo {volume} {107}},\ \bibinfo {pages} {056003} (\bibinfo {year} {2023})},\ \Eprint {http://arxiv.org/abs/2211.09459} {arXiv:2211.09459 [nucl-th]} \BibitemShut {NoStop}%
\bibitem [{\citenamefont {Gavassino}\ and\ \citenamefont {Shokri}(2023)}]{GavassinoShokryMHD:2023qnw}%
  \BibitemOpen
  \bibfield  {author} {\bibinfo {author} {\bibfnamefont {L.}~\bibnamefont {Gavassino}}\ and\ \bibinfo {author} {\bibfnamefont {M.}~\bibnamefont {Shokri}},\ }\href {\doibase 10.1103/PhysRevD.108.096010} {\bibfield  {journal} {\bibinfo  {journal} {Phys. Rev. D}\ }\textbf {\bibinfo {volume} {108}},\ \bibinfo {pages} {096010} (\bibinfo {year} {2023})},\ \Eprint {http://arxiv.org/abs/2307.11615} {arXiv:2307.11615 [nucl-th]} \BibitemShut {NoStop}%
\bibitem [{\citenamefont {Hiscock}\ and\ \citenamefont {Lindblom}(1985)}]{Hiscock_Insatibility_first_order}%
  \BibitemOpen
  \bibfield  {author} {\bibinfo {author} {\bibfnamefont {W.}~\bibnamefont {Hiscock}}\ and\ \bibinfo {author} {\bibfnamefont {L.}~\bibnamefont {Lindblom}},\ }\href {\doibase 10.1103/PhysRevD.31.725} {\bibfield  {journal} {\bibinfo  {journal} {Physical review D: Particles and fields}\ }\textbf {\bibinfo {volume} {31}},\ \bibinfo {pages} {725} (\bibinfo {year} {1985})}\BibitemShut {NoStop}%
\bibitem [{\citenamefont {{Kost{\"a}dt}}\ and\ \citenamefont {{Liu}}(2000)}]{Kost2000}%
  \BibitemOpen
  \bibfield  {author} {\bibinfo {author} {\bibfnamefont {P.}~\bibnamefont {{Kost{\"a}dt}}}\ and\ \bibinfo {author} {\bibfnamefont {M.}~\bibnamefont {{Liu}}},\ }\href {\doibase 10.1103/PhysRevD.62.023003} {\bibfield  {journal} {\bibinfo  {journal} {\prd}\ }\textbf {\bibinfo {volume} {62}},\ \bibinfo {eid} {023003} (\bibinfo {year} {2000})},\ \Eprint {http://arxiv.org/abs/cond-mat/0010276} {arXiv:cond-mat/0010276 [cond-mat.stat-mech]} \BibitemShut {NoStop}%
\bibitem [{\citenamefont {Gavassino}\ \emph {et~al.}(2020)\citenamefont {Gavassino}, \citenamefont {Antonelli},\ and\ \citenamefont {Haskell}}]{GavassinoLyapunov_2020}%
  \BibitemOpen
  \bibfield  {author} {\bibinfo {author} {\bibfnamefont {L.}~\bibnamefont {Gavassino}}, \bibinfo {author} {\bibfnamefont {M.}~\bibnamefont {Antonelli}}, \ and\ \bibinfo {author} {\bibfnamefont {B.}~\bibnamefont {Haskell}},\ }\href {\doibase 10.1103/physrevd.102.043018} {\bibfield  {journal} {\bibinfo  {journal} {Physical Review D}\ }\textbf {\bibinfo {volume} {102}} (\bibinfo {year} {2020}),\ 10.1103/physrevd.102.043018}\BibitemShut {NoStop}%
\bibitem [{\citenamefont {Gavassino}(2022)}]{GavassinoSuperlum2021}%
  \BibitemOpen
  \bibfield  {author} {\bibinfo {author} {\bibfnamefont {L.}~\bibnamefont {Gavassino}},\ }\href {\doibase 10.1103/PhysRevX.12.041001} {\bibfield  {journal} {\bibinfo  {journal} {Phys. Rev. X}\ }\textbf {\bibinfo {volume} {12}},\ \bibinfo {pages} {041001} (\bibinfo {year} {2022})},\ \Eprint {http://arxiv.org/abs/2111.05254} {arXiv:2111.05254 [gr-qc]} \BibitemShut {NoStop}%
\bibitem [{\citenamefont {Gavassino}\ and\ \citenamefont {Antonelli}(2025)}]{GavassinoAntonelli:2025umq}%
  \BibitemOpen
  \bibfield  {author} {\bibinfo {author} {\bibfnamefont {L.}~\bibnamefont {Gavassino}}\ and\ \bibinfo {author} {\bibfnamefont {M.}~\bibnamefont {Antonelli}},\ }\href@noop {} {\  (\bibinfo {year} {2025})},\ \Eprint {http://arxiv.org/abs/2509.00198} {arXiv:2509.00198 [gr-qc]} \BibitemShut {NoStop}%
\bibitem [{\citenamefont {Bemfica}\ \emph {et~al.}(2019)\citenamefont {Bemfica}, \citenamefont {Disconzi},\ and\ \citenamefont {Noronha}}]{Bemfica2019_conformal1}%
  \BibitemOpen
  \bibfield  {author} {\bibinfo {author} {\bibfnamefont {F.~S.}\ \bibnamefont {Bemfica}}, \bibinfo {author} {\bibfnamefont {M.~M.}\ \bibnamefont {Disconzi}}, \ and\ \bibinfo {author} {\bibfnamefont {J.}~\bibnamefont {Noronha}},\ }\href {\doibase 10.1103/PhysRevD.100.104020} {\bibfield  {journal} {\bibinfo  {journal} {Phys. Rev. D}\ }\textbf {\bibinfo {volume} {100}},\ \bibinfo {pages} {104020} (\bibinfo {year} {2019})}\BibitemShut {NoStop}%
\bibitem [{\citenamefont {{Kovtun}}(2019)}]{Kovtun2019}%
  \BibitemOpen
  \bibfield  {author} {\bibinfo {author} {\bibfnamefont {P.}~\bibnamefont {{Kovtun}}},\ }\href {\doibase 10.1007/JHEP10(2019)034} {\bibfield  {journal} {\bibinfo  {journal} {Journal of High Energy Physics}\ }\textbf {\bibinfo {volume} {2019}},\ \bibinfo {eid} {34} (\bibinfo {year} {2019})},\ \Eprint {http://arxiv.org/abs/1907.08191} {arXiv:1907.08191 [hep-th]} \BibitemShut {NoStop}%
\bibitem [{\citenamefont {Bemfica}\ \emph {et~al.}(2022)\citenamefont {Bemfica}, \citenamefont {Disconzi},\ and\ \citenamefont {Noronha}}]{BemficaDNDefinitivo2020}%
  \BibitemOpen
  \bibfield  {author} {\bibinfo {author} {\bibfnamefont {F.~S.}\ \bibnamefont {Bemfica}}, \bibinfo {author} {\bibfnamefont {M.~M.}\ \bibnamefont {Disconzi}}, \ and\ \bibinfo {author} {\bibfnamefont {J.}~\bibnamefont {Noronha}},\ }\href {\doibase 10.1103/PhysRevX.12.021044} {\bibfield  {journal} {\bibinfo  {journal} {Phys. Rev. X}\ }\textbf {\bibinfo {volume} {12}},\ \bibinfo {pages} {021044} (\bibinfo {year} {2022})}\BibitemShut {NoStop}%
\bibitem [{\citenamefont {Srednicki}(2007)}]{Srednicki_2007}%
  \BibitemOpen
  \bibfield  {author} {\bibinfo {author} {\bibfnamefont {M.}~\bibnamefont {Srednicki}},\ }\href@noop {} {\emph {\bibinfo {title} {Quantum Field Theory}}}\ (\bibinfo  {publisher} {Cambridge University Press},\ \bibinfo {year} {2007})\BibitemShut {NoStop}%
\bibitem [{\citenamefont {{Born}}(1909)}]{Born1909}%
  \BibitemOpen
  \bibfield  {author} {\bibinfo {author} {\bibfnamefont {M.}~\bibnamefont {{Born}}},\ }\href {\doibase 10.1002/andp.19093351102} {\bibfield  {journal} {\bibinfo  {journal} {Annalen der Physik}\ }\textbf {\bibinfo {volume} {335}},\ \bibinfo {pages} {1} (\bibinfo {year} {1909})}\BibitemShut {NoStop}%
\bibitem [{\citenamefont {{Herglotz}}(1910)}]{Herglotz1910}%
  \BibitemOpen
  \bibfield  {author} {\bibinfo {author} {\bibfnamefont {G.}~\bibnamefont {{Herglotz}}},\ }\href {\doibase 10.1002/andp.19103360208} {\bibfield  {journal} {\bibinfo  {journal} {Annalen der Physik}\ }\textbf {\bibinfo {volume} {336}},\ \bibinfo {pages} {393} (\bibinfo {year} {1910})}\BibitemShut {NoStop}%
\bibitem [{\citenamefont {{Noether}}(1910)}]{Noether1910}%
  \BibitemOpen
  \bibfield  {author} {\bibinfo {author} {\bibfnamefont {F.}~\bibnamefont {{Noether}}},\ }\href {\doibase 10.1002/andp.19103360504} {\bibfield  {journal} {\bibinfo  {journal} {Annalen der Physik}\ }\textbf {\bibinfo {volume} {336}},\ \bibinfo {pages} {919} (\bibinfo {year} {1910})}\BibitemShut {NoStop}%
\bibitem [{\citenamefont {Gourgoulhon}(2013)}]{special_in_gen}%
  \BibitemOpen
  \bibfield  {author} {\bibinfo {author} {\bibfnamefont {E.}~\bibnamefont {Gourgoulhon}},\ }\href@noop {} {\emph {\bibinfo {title} {Special Relativity in General Frames: From Particles to Astrophysics}}},\ \bibinfo {edition} {1st}\ ed.,\ Graduate Texts in Physics\ (\bibinfo  {publisher} {Springer-Verlag Berlin Heidelberg},\ \bibinfo {year} {2013})\BibitemShut {NoStop}%
\bibitem [{\citenamefont {Kovtun}(2023)}]{KovtunTemperature2022vas}%
  \BibitemOpen
  \bibfield  {author} {\bibinfo {author} {\bibfnamefont {P.}~\bibnamefont {Kovtun}},\ }\href {\doibase 10.1103/PhysRevD.107.086012} {\bibfield  {journal} {\bibinfo  {journal} {Phys. Rev. D}\ }\textbf {\bibinfo {volume} {107}},\ \bibinfo {pages} {086012} (\bibinfo {year} {2023})},\ \Eprint {http://arxiv.org/abs/2210.15605} {arXiv:2210.15605 [gr-qc]} \BibitemShut {NoStop}%
\bibitem [{\citenamefont {{Weinberg}}(1971)}]{Weinberg1971}%
  \BibitemOpen
  \bibfield  {author} {\bibinfo {author} {\bibfnamefont {S.}~\bibnamefont {{Weinberg}}},\ }\href {\doibase 10.1086/151073} {\bibfield  {journal} {\bibinfo  {journal} {\apj}\ }\textbf {\bibinfo {volume} {168}},\ \bibinfo {pages} {175} (\bibinfo {year} {1971})}\BibitemShut {NoStop}%
\bibitem [{\citenamefont {Gavassino}(2024)}]{GavassinoInfiniteOrded2024pgl}%
  \BibitemOpen
  \bibfield  {author} {\bibinfo {author} {\bibfnamefont {L.}~\bibnamefont {Gavassino}},\ }\href {\doibase 10.1103/PhysRevLett.133.032302} {\bibfield  {journal} {\bibinfo  {journal} {Phys. Rev. Lett.}\ }\textbf {\bibinfo {volume} {133}},\ \bibinfo {pages} {032302} (\bibinfo {year} {2024})},\ \Eprint {http://arxiv.org/abs/2402.19343} {arXiv:2402.19343 [nucl-th]} \BibitemShut {NoStop}%
\bibitem [{\citenamefont {Wald}(1984)}]{Wald}%
  \BibitemOpen
  \bibfield  {author} {\bibinfo {author} {\bibfnamefont {R.~M.}\ \bibnamefont {Wald}},\ }\href {https://cds.cern.ch/record/106274} {\emph {\bibinfo {title} {{General relativity}}}}\ (\bibinfo  {publisher} {Chicago Univ. Press},\ \bibinfo {address} {Chicago, IL},\ \bibinfo {year} {1984})\BibitemShut {NoStop}%
\bibitem [{\citenamefont {{Misner}}\ \emph {et~al.}(1973)\citenamefont {{Misner}}, \citenamefont {{Thorne}},\ and\ \citenamefont {{Wheeler}}}]{MTW_book}%
  \BibitemOpen
  \bibfield  {author} {\bibinfo {author} {\bibfnamefont {C.~W.}\ \bibnamefont {{Misner}}}, \bibinfo {author} {\bibfnamefont {K.~S.}\ \bibnamefont {{Thorne}}}, \ and\ \bibinfo {author} {\bibfnamefont {J.~A.}\ \bibnamefont {{Wheeler}}},\ }\href@noop {} {\emph {\bibinfo {title} {San Francisco: W.H.~Freeman and Co., 1973}}}\ (\bibinfo {year} {1973})\BibitemShut {NoStop}%
\bibitem [{\citenamefont {Israel}(1981)}]{Israel1981}%
  \BibitemOpen
  \bibfield  {author} {\bibinfo {author} {\bibfnamefont {W.}~\bibnamefont {Israel}},\ }\href {\doibase https://doi.org/10.1016/0378-4371(81)90220-X} {\bibfield  {journal} {\bibinfo  {journal} {Physica A: Statistical Mechanics and its Applications}\ }\textbf {\bibinfo {volume} {106}},\ \bibinfo {pages} {204} (\bibinfo {year} {1981})}\BibitemShut {NoStop}%
\bibitem [{\citenamefont {Callen}(1985)}]{Callen_book}%
  \BibitemOpen
  \bibfield  {author} {\bibinfo {author} {\bibfnamefont {H.~B.}\ \bibnamefont {Callen}},\ }\href {https://cds.cern.ch/record/450289} {\emph {\bibinfo {title} {{Thermodynamics and an introduction to thermostatistics; 2nd ed.}}}}\ (\bibinfo  {publisher} {Wiley},\ \bibinfo {address} {New York, NY},\ \bibinfo {year} {1985})\BibitemShut {NoStop}%
\bibitem [{\citenamefont {Bajec}\ and\ \citenamefont {Soloviev}(2025)}]{Bajec:2025dqm}%
  \BibitemOpen
  \bibfield  {author} {\bibinfo {author} {\bibfnamefont {M.}~\bibnamefont {Bajec}}\ and\ \bibinfo {author} {\bibfnamefont {A.}~\bibnamefont {Soloviev}},\ }\href@noop {} {\  (\bibinfo {year} {2025})},\ \Eprint {http://arxiv.org/abs/2506.15531} {arXiv:2506.15531 [hep-th]} \BibitemShut {NoStop}%
\bibitem [{\citenamefont {Gavassino}(2023)}]{GavassinoBounds2023myj}%
  \BibitemOpen
  \bibfield  {author} {\bibinfo {author} {\bibfnamefont {L.}~\bibnamefont {Gavassino}},\ }\href {\doibase 10.1016/j.physletb.2023.137854} {\bibfield  {journal} {\bibinfo  {journal} {Phys. Lett. B}\ }\textbf {\bibinfo {volume} {840}},\ \bibinfo {pages} {137854} (\bibinfo {year} {2023})},\ \Eprint {http://arxiv.org/abs/2301.06651} {arXiv:2301.06651 [hep-th]} \BibitemShut {NoStop}%
\bibitem [{\citenamefont {Thomas}(1927)}]{Thomas_1927}%
  \BibitemOpen
  \bibfield  {author} {\bibinfo {author} {\bibfnamefont {L.~H.}\ \bibnamefont {Thomas}},\ }\href {\doibase 10.1017/S0305004100011683} {\bibfield  {journal} {\bibinfo  {journal} {Mathematical Proceedings of the Cambridge Philosophical Society}\ }\textbf {\bibinfo {volume} {23}},\ \bibinfo {pages} {542–548} (\bibinfo {year} {1927})}\BibitemShut {NoStop}%
\bibitem [{\citenamefont {Fermi}(1927)}]{Fermi1927}%
  \BibitemOpen
  \bibfield  {author} {\bibinfo {author} {\bibfnamefont {E.}~\bibnamefont {Fermi}},\ }\href@noop {} {\bibfield  {journal} {\bibinfo  {journal} {Rendiconti Accademia Nazionale dei Lincei}\ }\textbf {\bibinfo {volume} {6}},\ \bibinfo {pages} {602} (\bibinfo {year} {1927})}\BibitemShut {NoStop}%
\bibitem [{\citenamefont {{Geppert}}\ and\ \citenamefont {{Wiebicke}}(1991)}]{Geppert1991}%
  \BibitemOpen
  \bibfield  {author} {\bibinfo {author} {\bibfnamefont {U.}~\bibnamefont {{Geppert}}}\ and\ \bibinfo {author} {\bibfnamefont {H.~J.}\ \bibnamefont {{Wiebicke}}},\ }\href@noop {} {\bibfield  {journal} {\bibinfo  {journal} {\aaps}\ }\textbf {\bibinfo {volume} {87}},\ \bibinfo {pages} {217} (\bibinfo {year} {1991})}\BibitemShut {NoStop}%
\bibitem [{\citenamefont {{Yakovlev}}(1984)}]{Yakovlev1984}%
  \BibitemOpen
  \bibfield  {author} {\bibinfo {author} {\bibfnamefont {D.~G.}\ \bibnamefont {{Yakovlev}}},\ }\href {\doibase 10.1007/BF00651950} {\bibfield  {journal} {\bibinfo  {journal} {\apss}\ }\textbf {\bibinfo {volume} {98}},\ \bibinfo {pages} {37} (\bibinfo {year} {1984})}\BibitemShut {NoStop}%
\bibitem [{\citenamefont {Schmitt}\ and\ \citenamefont {Shternin}(2018)}]{Schmitt:2017efp}%
  \BibitemOpen
  \bibfield  {author} {\bibinfo {author} {\bibfnamefont {A.}~\bibnamefont {Schmitt}}\ and\ \bibinfo {author} {\bibfnamefont {P.}~\bibnamefont {Shternin}},\ }\href {\doibase 10.1007/978-3-319-97616-7_9} {\bibfield  {journal} {\bibinfo  {journal} {Astrophys. Space Sci. Libr.}\ }\textbf {\bibinfo {volume} {457}},\ \bibinfo {pages} {455} (\bibinfo {year} {2018})},\ \Eprint {http://arxiv.org/abs/1711.06520} {arXiv:1711.06520 [astro-ph.HE]} \BibitemShut {NoStop}%
\bibitem [{\citenamefont {Harutyunyan}\ and\ \citenamefont {Sedrakian}(2024)}]{Harutyunyan:2024hsd}%
  \BibitemOpen
  \bibfield  {author} {\bibinfo {author} {\bibfnamefont {A.}~\bibnamefont {Harutyunyan}}\ and\ \bibinfo {author} {\bibfnamefont {A.}~\bibnamefont {Sedrakian}},\ }\href {\doibase 10.3390/particles7040059} {\bibfield  {journal} {\bibinfo  {journal} {Particles}\ }\textbf {\bibinfo {volume} {7}},\ \bibinfo {pages} {967} (\bibinfo {year} {2024})},\ \Eprint {http://arxiv.org/abs/2409.01304} {arXiv:2409.01304 [astro-ph.HE]} \BibitemShut {NoStop}%
\bibitem [{\citenamefont {Gakis}\ and\ \citenamefont {Gourgouliatos}(2024)}]{Gakis:2024obw}%
  \BibitemOpen
  \bibfield  {author} {\bibinfo {author} {\bibfnamefont {D.}~\bibnamefont {Gakis}}\ and\ \bibinfo {author} {\bibfnamefont {K.~N.}\ \bibnamefont {Gourgouliatos}},\ }\href {\doibase 10.1051/0004-6361/202449692} {\bibfield  {journal} {\bibinfo  {journal} {Astron. Astrophys.}\ }\textbf {\bibinfo {volume} {690}},\ \bibinfo {pages} {A117} (\bibinfo {year} {2024})},\ \Eprint {http://arxiv.org/abs/2402.14911} {arXiv:2402.14911 [astro-ph.HE]} \BibitemShut {NoStop}%
\bibitem [{\citenamefont {{Weinberg}}(1972)}]{Weinberg_book_1972}%
  \BibitemOpen
  \bibfield  {author} {\bibinfo {author} {\bibfnamefont {S.}~\bibnamefont {{Weinberg}}},\ }\href@noop {} {\emph {\bibinfo {title} {{Gravitation and Cosmology: Principles and Applications of the General Theory of Relativity}}}}\ (\bibinfo {year} {1972})\BibitemShut {NoStop}%
\bibitem [{\citenamefont {Gavassino}(2026)}]{GavassinoPlasmaOscillations:2025tul}%
  \BibitemOpen
  \bibfield  {author} {\bibinfo {author} {\bibfnamefont {L.}~\bibnamefont {Gavassino}},\ }\href {\doibase 10.1103/b8yz-9rj7} {\bibfield  {journal} {\bibinfo  {journal} {Phys. Rev. D}\ }\textbf {\bibinfo {volume} {113}},\ \bibinfo {pages} {023035} (\bibinfo {year} {2026})},\ \Eprint {http://arxiv.org/abs/2511.14344} {arXiv:2511.14344 [physics.plasm-ph]} \BibitemShut {NoStop}%
\bibitem [{\citenamefont {Eckart}(1940)}]{Eckart40}%
  \BibitemOpen
  \bibfield  {author} {\bibinfo {author} {\bibfnamefont {C.}~\bibnamefont {Eckart}},\ }\href {\doibase 10.1103/PhysRev.58.919} {\bibfield  {journal} {\bibinfo  {journal} {Phys. Rev.}\ }\textbf {\bibinfo {volume} {58}},\ \bibinfo {pages} {919} (\bibinfo {year} {1940})}\BibitemShut {NoStop}%
\bibitem [{\citenamefont {Aguilera}\ \emph {et~al.}(2008)\citenamefont {Aguilera}, \citenamefont {Pons},\ and\ \citenamefont {Miralles}}]{Aguilera:2007xk}%
  \BibitemOpen
  \bibfield  {author} {\bibinfo {author} {\bibfnamefont {D.~N.}\ \bibnamefont {Aguilera}}, \bibinfo {author} {\bibfnamefont {J.~A.}\ \bibnamefont {Pons}}, \ and\ \bibinfo {author} {\bibfnamefont {J.~A.}\ \bibnamefont {Miralles}},\ }\href {\doibase 10.1051/0004-6361:20078786} {\bibfield  {journal} {\bibinfo  {journal} {Astron. Astrophys.}\ }\textbf {\bibinfo {volume} {486}},\ \bibinfo {pages} {255} (\bibinfo {year} {2008})},\ \Eprint {http://arxiv.org/abs/0710.0854} {arXiv:0710.0854 [astro-ph]} \BibitemShut {NoStop}%
\bibitem [{\citenamefont {Chandra}\ \emph {et~al.}(2015)\citenamefont {Chandra}, \citenamefont {Gammie}, \citenamefont {Foucart},\ and\ \citenamefont {Quataert}}]{Chandra:2015iza}%
  \BibitemOpen
  \bibfield  {author} {\bibinfo {author} {\bibfnamefont {M.}~\bibnamefont {Chandra}}, \bibinfo {author} {\bibfnamefont {C.~F.}\ \bibnamefont {Gammie}}, \bibinfo {author} {\bibfnamefont {F.}~\bibnamefont {Foucart}}, \ and\ \bibinfo {author} {\bibfnamefont {E.}~\bibnamefont {Quataert}},\ }\href {\doibase 10.1088/0004-637X/810/2/162} {\bibfield  {journal} {\bibinfo  {journal} {Astrophys. J.}\ }\textbf {\bibinfo {volume} {810}},\ \bibinfo {pages} {162} (\bibinfo {year} {2015})},\ \Eprint {http://arxiv.org/abs/1508.00878} {arXiv:1508.00878 [astro-ph.HE]} \BibitemShut {NoStop}%
\end{thebibliography}%

\newpage

\onecolumngrid
\newpage
\begin{center}
\textbf{\large Thermoelectric conduction in General Relativity: a causal, stable, and well-posed theory\\Supplementary Material}\\[.2cm]
  L. Gavassino\\[.1cm]
  {\itshape Department of Applied Mathematics and Theoretical Physics, University of Cambridge, Wilberforce Road, Cambridge CB3 0WA, United Kingdom\\}
\end{center}

\setcounter{equation}{0}
\setcounter{figure}{0}
\setcounter{table}{0}
\setcounter{page}{1}
\renewcommand{\theequation}{S\arabic{equation}}
\renewcommand{\thefigure}{S\arabic{figure}}

\section*{Stability Theorem}

Here, we prove the following theorem.
\begin{theorem}
Let $\Pi$ and $\Lambda$ be two Hermitian non-negative definite matrices of the same size, and let $k$ be a real number. Then, the solutions of
\begin{equation}\label{theor2one}
[\Gamma^2+(\Pi{+}\Lambda) \Gamma +\Pi \Lambda+k^2]\Upsilon =0 \spc (\text{with } \Upsilon\neq 0, \, \, \Gamma\in \mathbb{C})
\end{equation}
are such that $\mathfrak{Re}\Gamma\leq 0$.
\end{theorem}
\begin{proof}
Let us first note that equation \eqref{theor2one} can be rearranged as follows:
\begin{equation}\label{bobrego}
[(\Gamma {+}\Pi)(\Gamma {+}\Lambda)+k^2]\Upsilon =0\, .
\end{equation}
Define $\Tilde{\Upsilon}=(\Gamma{+}\Lambda)\Upsilon$, and multiply both sides of \eqref{bobrego} by $\Tilde{\Upsilon}^\dagger=\Upsilon^\dagger (\Gamma^*{+}\Lambda)$, which gives
\begin{equation}\label{thepro}
\Tilde{\Upsilon}^\dagger (\Gamma {+}\Pi)\Tilde{\Upsilon}+k^2\Upsilon^\dagger (\Gamma^* {+}\Lambda)\Upsilon =0\, .   
\end{equation}
Now, let us note that $\Tilde{\Upsilon}^\dagger\Tilde{\Upsilon}$, $\Tilde{\Upsilon}^\dagger\Pi \Tilde{\Upsilon}$, $\Upsilon^\dagger \Upsilon$, and $\Upsilon^\dagger \Lambda \Upsilon$ are all real and non-negative. Hence, if we take the real part of \eqref{thepro}, and isolate $\mathfrak{Re}\Gamma$, we obtain
\begin{equation}
\mathfrak{Re}\Gamma= -\dfrac{\Tilde{\Upsilon}^\dagger \Pi\Tilde{\Upsilon}+k^2 \Upsilon^\dagger \Lambda \Upsilon}{\Tilde{\Upsilon}^\dagger\Tilde{\Upsilon}+k^2 \Upsilon^\dagger \Upsilon} \leq 0\, ,
\end{equation}
which is what we wanted to prove.
\end{proof} 

\noindent\textbf{Remark:} Note that we did not specify the dimension of the matrices. Hence, the stability result would remain valid even if we added additional chemical potentials that undergo diffusion and couple to $\mu$ and $T$.\\



\newpage
\section*{Dispersion relations of the theory: Structure and Interpretation}

We have shown that our theory, when linearized around constant (charge-neutral) states, is stable. However, we did not yet analyze the dispersion relations of the quasinormal modes in detail. Here, we do it.

\vspace{-0.1cm}
\subsection{Global mode structure}
\vspace{-0.1cm}

Let us first analyze the transversal modes propagating in the $z$ direction, which are governed by the equations
\begin{equation}\label{cattenoB}
\partial_t^2 \mathcal{B}_x +\sigma_1 \partial_t \mathcal{B}_x -\partial^2_z \mathcal{B}_x =0 \, , \spc  \partial_t^2 \mathcal{B}_y +\sigma_1 \partial_t \mathcal{B}_y -\partial^2_z \mathcal{B}_y =0 \, .
\end{equation}
Under a plane wave assumption ($\delta\Psi\propto e^{ikx-i\omega t}$), we obtain the relation $\omega^2+i\sigma_1\omega-k^2=0$, whose solutions are plotted in figure \ref{fig:transmodes}, left panel. The interpretation is standard: At high $k$, we have 4 distinct electromagnetic waves (two left-moving polarizations and two right-moving polarizations) that travel at speed $1$, and decay over a timescale $2/\sigma_1$. However, when the wavelength $k^{-1}$ becomes longer than $2/\sigma_1$, the waves stop propagating, and they just decay. In the limit $k\rightarrow 0$, we have two diffusive modes (one for each transversal component of $\mathcal{B}$) with $\omega=-ik^2/\sigma_1$, and two non-hydrodynamic modes that decay over a timescale $1/\sigma_1$.

Let us now turn our attention to the 4 longitudinal modes. These modes describe the relaxation of the degrees of freedom $\{\delta T,\partial_t \delta T,\delta \mu,\partial_t \delta \mu\}$. Their dispersion relations $\omega(k)$ are solutions of the equation
\begin{equation}\label{deTTo}
\det[(-i\omega {+}\Pi)(-i\omega {+}\Lambda)+k^2]
= 0 \, , \spc \text{with }
\Lambda=
\begin{bmatrix}
\lambda_1 & 0 \\
0 & \lambda_2 \\
\end{bmatrix}\, ,
\quad \Pi =
\begin{bmatrix}
A & B \\
B & C \\
\end{bmatrix}\, , 
\end{equation}
where $\lambda_1$, $\lambda_2$, $A$, and $C$ are non-negative, and $B=\pm \sqrt{AC}$ (since $\det \Pi=0$). 
Since the polynomial is of $4^{\text{th}}$ degree in $\omega$, there is no simple, universal, mode structure. Instead, many mode configurations are possible, depending on the values of the coefficients. The plot in figure \ref{fig:transmodes} (right panel) is just one example. The asymptotics are, however, universal. In particular, at high $k$, we always have two right-moving waves and two left-moving waves, all of which propagate at speed $1$. At small $k$, there are always three gapped modes, and only one gapless mode. This follows from the fact that \eqref{deTTo} becomes $\det(-i\omega+\Pi)\det(-i\omega+\Lambda)=0$, whose roots are $\omega=-i\lambda_{1/2}$ (both of which are non-zero) and $\omega=-i\times \text{``eigenvalues of }\Pi\text{''}$, one of which is always zero (since $\det\Pi=0$) and the other of which is always $\sigma_1$ (since $\text{Tr}\, \Pi=\sigma_1$). The above analysis also reveals that, at $k=0$, $\omega$ is always imaginary (since $\Pi$ and $\Lambda$ are symmetric). This differentiates our theory from the Israel-Stewart-Maxwell model \cite{Dash:2022xkz,GavassinoShokryMHD:2023qnw} and from kinetic theory, which were found to undergo electrically-driven (plasma) oscillations in certain regimes at zero wavenumber \cite{,GavassinoPlasmaOscillations:2025tul}.

\begin{figure}[b!]
    \centering
\includegraphics[width=0.4\linewidth]{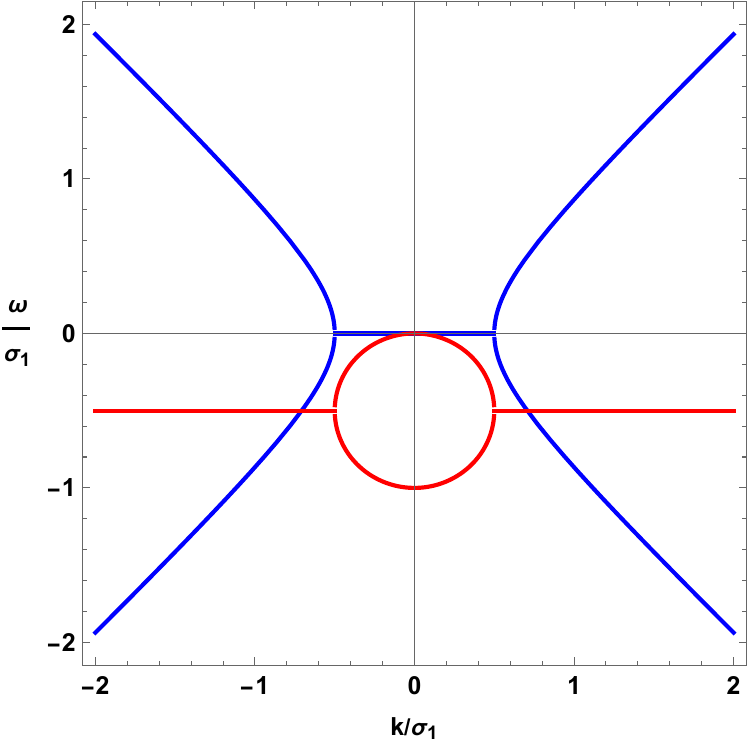}
\includegraphics[width=0.4\linewidth]{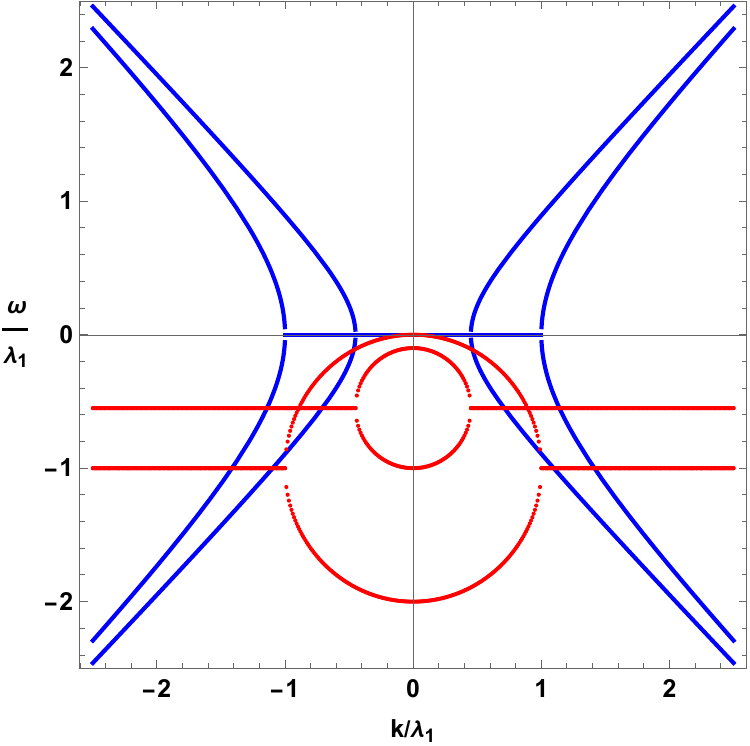}
\caption{Dispersion relations of the transversal modes (left panel) and of the longitudinal modes (right panel) of the theory. The blue lines are $\mathfrak{Re}\,\omega$, and the red lines are $\mathfrak{Im}\,\omega$. Note that the structure of the transversal modes is universal, while the structure of the longitudinal modes (e.g. the shape, relative size, and relative position of the red circles) may vary depending on the parameters.}
    \label{fig:transmodes}
\end{figure}

\subsection{The gapless mode describes heat transport}
\vspace{-0.1cm}

We found that there exists a single longitudinal gapless mode. Let us now examine its physical interpretation. 

Let $\{\Gamma(k),\Upsilon(k)\}$ be the eigen-rate and the eigen-solution that describe the diffusive mode, and differentiate the equation of motion \eqref{theor2one} along such eigen-couple. The zeroth, first, and second derivatives of the equation are
\begin{equation}
\begin{split}
& [\Gamma^2+(\Pi{+}\Lambda) \Gamma +\Pi \Lambda+k^2]\Upsilon =0 \, ,\\
& [\Gamma^2+(\Pi{+}\Lambda) \Gamma +\Pi \Lambda+k^2]\Upsilon' +[2\Gamma\Gamma'+(\Pi{+}\Lambda) \Gamma'+2k]\Upsilon =0\, ,\\
& [\Gamma^2+(\Pi{+}\Lambda) \Gamma +\Pi \Lambda+k^2]\Upsilon'' +2[2\Gamma\Gamma'+(\Pi{+}\Lambda) \Gamma'+2k]\Upsilon'+[2\Gamma\Gamma''+2(\Gamma')^2+(\Pi{+}\Lambda) \Gamma''+2]\Upsilon =0\, .\\
\end{split}
\end{equation}
Evaluating the above equations at $k=0$, and considering that $\Gamma(0)=\Gamma'(0)=0$ (the mode is diffusive), we obtain
\begin{equation}\label{zerando}
\begin{split}
& \Pi \Lambda\Upsilon(0) =0 \, ,\\
& \Pi \Lambda\Upsilon'(0) =0\, ,\\
& \Pi \Lambda\Upsilon''(0) +[(\Pi{+}\Lambda) \Gamma''(0)+2]\Upsilon(0) =0\, .\\
\end{split}
\end{equation}
Now, let us note that, if $\Pi\Lambda\Upsilon(0)$ vanishes, then also $\Upsilon^\dagger(0) \Lambda\Pi$ vanishes. Hence, we can multiply the third line of \eqref{zerando} by $\Upsilon^\dagger \Lambda$, and we obtain
\begin{equation}\label{fgg}
\dfrac{1}{2}\Gamma''(0)=-\dfrac{\Upsilon^\dagger(0) \Lambda \Upsilon(0)}{\Upsilon^\dagger(0) \Lambda^2 \Upsilon(0)}\, .
\end{equation}
These results are not particularly enlightening because they are expressed in terms of $\Upsilon$. If we express them in terms of $(\delta\mu,\delta T)$, we learn much more. In particular, the statement $\Pi \Lambda\Upsilon(0) =0$ becomes $\delta \rho_\text{lte}(0)=0$, which tells us that the diffusive mode describes pure heat transfer (as it carries no electric charge), while \eqref{fgg} becomes
\begin{equation}
\omega=-i\dfrac{\partial T}{\partial s}\bigg|_\rho \left( \sigma_5 -\dfrac{\sigma_3^2}{\sigma_1}\right)k^2+\mathcal{O}(k^4) \, ,
\end{equation}
which tells that the combination $T\sigma_5-T\sigma_3^2/\sigma_1(\equiv\kappa_3-\mu\sigma_3-T\sigma_3^2/\sigma_1)$ may be interpreted as the heat conductivity.

\vspace{-0.1cm}
\subsection{Origin of the gap of the charge-diffusion mode}
\vspace{-0.1cm}

Ordinarily, a conserved current gives rise to a gapless diffusive mode. The electric current, however, is an exception. To see this, start from the conservation law $\partial_t\delta J^0+\partial_z \delta J^z=0$, and take the long-wavelength limit ($k\to 0$). In this regime, the constitutive relation for the current reduces to $\delta J^z \approx \sigma_1 \mathcal{E}^z$, so the charge conservation equation becomes $0=\partial_t\delta J^0+\sigma_1\partial_z \mathcal{E}^z=\partial_t\delta J^0+\sigma_1 \delta J^0$.
Its solution is a non-hydrodynamic mode with a finite gap, $\omega(k{=}0)=-i\sigma_1$. As can be seen, the origin of this gap is the Maxwell equation $\partial_j \mathcal{E}^j=\delta J^0$, which causes the electric field (and thus also the current) to blow up like $1/k$ in the infrared limit.

\vspace{-0.1cm}
\subsection{Which modes are physical?}
\vspace{-0.1cm}

Since our theory is obtained through a formal expansion in $\tau$, its solutions can be trusted only if their Fourier support remains sufficiently close to the origin, namely when $\tau |k|\ll 1$ and $\tau |\omega|\ll 1$. This provides a simple criterion to distinguish between the physical modes (i.e. those that fall within the validity regime of theory) and the unphysical ones. Let us now examine what this criterion implies.

Consider the left panel of figure \ref{fig:transmodes}, which depicts the transverse modes. The red circle has a diameter of $\sigma_1$, along both the $k$ and $\omega$ axes. To compare this scale with $\tau$, we recall the Drude relation $\sigma_1 = n_e e^2 \tau/m$, where $n_e$ and $m$ denote the electron density and mass. From this, we obtain the estimate
\begin{equation}\label{tausigma}
\tau \sigma_1 \sim \frac{n_e e^2}{m} \tau^2 = \omega_p^2 \tau^2 \, ,
\end{equation}
with $\omega_p = \sqrt{n_e e^2/m}$ the plasma frequency.
Thus, if the medium satisfies $\omega_p \tau \ll 1$, both the hydrodynamic mode (upper half of the circle) and the non-hydrodynamic mode (lower half of the circle) lie within the domain of validity of the first-order theory and are therefore physical. In this regime, the electronic degrees of freedom are strongly damped, and the coupled Maxwell-matter system is accurately captured by the present framework, even at frequencies $\omega \gtrsim \sigma_1$ (compare with figure 4 of \cite{GavassinoPlasmaOscillations:2025tul}, left panel). Notably, our framework also correctly captures the propagation of electromagnetic waves with $k\gtrsim \sigma_1/2$.
\newpage

Conversely, when $\omega_p \tau \gg 1$, only a small portion of the hydrodynamic branch, close to the origin in frequency-momentum space, remains within the validity domain of the theory. This infrared sector is governed by an ordinary diffusion equation, $\partial_t \mathcal{B}_{x}\approx-\partial_z^2\mathcal{B}_{x}$
(and similarly for $\mathcal{B}_y$).
The remainder of the spectrum is instead characterized by underdamped plasma oscillations, which fall outside the applicability of a first-order description. Accurately modeling this regime (and, in particular, the propagating electromagnetic waves) requires the use of Israel-Stewart-type theories, which have been shown to correctly capture this kind of oscillatory dynamics \cite{GavassinoPlasmaOscillations:2025tul}.

Let us now focus on the longitudinal modes (right panel of figure \ref{fig:transmodes}). The uppermost point of the smaller circle is the charge-diffusion mode, which has a distance $\sigma_1$ from the origin. Hence, invoking again the estimate \eqref{tausigma}, we conclude that the charge-diffusion model is physical if $\omega_p\tau \ll 1$. As for the lowermost points of the two circles, these are located at distances $\lambda_1$ and $\lambda_2$ from the origin. A quick estimate then gives
\begin{equation}
\tau \lambda_i \sim \dfrac{\partial \rho}{\partial\mu} \dfrac{\tau}{\sigma_1}\sim \dfrac{1}{b^2 \omega_p^2}\sim \dfrac{1}{v_e^2}\geq 1 \, ,
\end{equation}
where $b$ is the Thomas-Fermi length and $v_e\leq 1$ the typical speed of the electrons. We conclude that the lower halves of the red circles always lie outside the regime of validity of the theory. This is expected, since these modes are artifacts of our choice of hydrodynamic frame (i.e. of how we define $T$ and $\mu$ out of equilibrium \cite{Kovtun2019}). In the Eckart frame \cite{Eckart40}, which is acausal (but is the conventional frame in non-relativistic theories \cite[\S 26]{landau8}\cite[\S 21.C]{SommerfeldStatmech}), such modes do not appear. Indeed, in the Newtonian limit ($v_e\ll 1$), these frame-dependent modes are pushed infinitely far from the origin, implying that they decay infinitely fast. In this way, we effectively recover the Eckart frame.
\newpage
\section*{Relationship between charge density and charge chemical potential in metals}

We consider a solid lattice of ions with atomic number $Z$, filled with freely moving electrons. The differential of the energy density $\varepsilon$ (dropping the ``lte'' subscripts) reads $d\varepsilon =T ds+ \mu_I dn_I +\mu_e dn_e$,
where $n_I$ and $n_e$ are the number densities of ions and electrons, while $\mu_I$ and $\mu_e$ are the respective chemical potentials. Since the lattice is assumed rigid and with uniform spacing, $n_I$ is effectively a constant (so $dn_I=0$). Moreover, the charge density is 
\begin{equation}\label{rhoe}
\rho=eZn_I-en_e  \, ,  
\end{equation}
and so we have $d\rho=-e dn_e$. Hence, we are left with $d\varepsilon =T ds -(\mu_e/e) d\rho$, which immediately yields 
\begin{equation}\label{mue}
\mu=-\mu_e/e\, .
\end{equation}
But at low temperatures, $\mu_e$ coincides with the Fermi energy of the electrons. Hence, working in the free electron model, which treats the electrons as a degenerate non-relativistic Fermi gas, we have
\begin{equation}\label{mueee}
\mu_e =m+\dfrac{1}{2m} (3\pi^2 n_e)^{2/3} \, ,
\end{equation}
with $m$ the mass of the electron. Combining \eqref{rhoe}, \eqref{mue} and \eqref{mueee}, we obtain 
\begin{equation}
\mu =-\dfrac{m}{e} -\dfrac{1}{2me} \left[3\pi^2 Zn_I\left( 1-\dfrac{\rho}{eZn_I}\right) \right]^{2/3}\, .
\end{equation}
Taylor expanding to first order in $\rho$, we recover the equation $\mu=-\mu_0+b^2\rho+\mathcal{O}(\rho^2)$, with 
\begin{equation}
\begin{split}
\mu_0 ={}& \dfrac{m}{e}+\dfrac{1}{2me} (3\pi^2 Zn_I)^{2/3} \, , \\
b^2={}& \dfrac{\pi^{4/3}}{me^2 (3Z n_I)^{1/3}}\, .
\end{split}
\end{equation}
It is immediate to verify that $b$ is just the usual Thomas-Fermi length scale in conductors.

\newpage
\section*{Strong magnetic fields}

In the presence of strong magnetic fields, nothing fundamental changes: the currents $\{\K J^\mu,W^\mu\}$ still contain gradient terms proportional to $\{\nabla_\nu(\K\mu){-}\K \mathcal{E}_\nu,\nabla_\nu(\K T)\}$. However, we can now use the field 
\begin{equation}
\mathcal{B}^\mu =-\dfrac{1}{2} \varepsilon^{\mu\alpha\beta\nu}F_{\alpha\beta}u_\nu 
\end{equation}
to build additional tensor structures. Specifically, we have the following constitutive relations ($\Hat{\mathcal{B}}^\mu=\mathcal{B}^\mu/\sqrt{\mathcal{B}^\alpha\mathcal{B}_\alpha}$):
\begin{equation}
\begin{split}
\K J^\mu {=}\rho_{\text{lte}} \K^\mu & {-}(\sigma_1 g^{\mu \nu}{+}\sigma_2 u^\mu u^\nu{+}\sigma_\mathcal{B} \Hat{\mathcal{B}}^\mu \Hat{\mathcal{B}}^\nu{+}\sigma_H \varepsilon^{\mu \nu \alpha\beta}u_\alpha\Hat{\mathcal{B}}_\beta) [\nabla_\nu (\K\mu){-}\K\mathcal{E}_\nu] \\
& {-}(\sigma_3 g^{\mu \nu}{+}\sigma_4 u^\mu u^\nu{+}\sigma'_\mathcal{B} \Hat{\mathcal{B}}^\mu \Hat{\mathcal{B}}^\nu{+}\sigma'_H \varepsilon^{\mu \nu \alpha\beta}u_\alpha\Hat{\mathcal{B}}_\beta) \nabla_\nu (\K T) {+}\mathcal{O}(\tau^2) , \\
W^\mu {=}\varepsilon_{\text{lte}} \K^\mu & {-}(\kappa_1 g^{\mu \nu}{+}\kappa_2 u^\mu u^\nu{+}\kappa_\mathcal{B} \Hat{\mathcal{B}}^\mu \Hat{\mathcal{B}}^\nu{+}\kappa_H \varepsilon^{\mu \nu \alpha\beta}u_\alpha\Hat{\mathcal{B}}_\beta) [\nabla_\nu (\K\mu){-}\K \mathcal{E}_\nu] \\
& {-}(\kappa_3 g^{\mu \nu}{+}\kappa_4 u^\mu u^\nu{+}\kappa'_\mathcal{B} \Hat{\mathcal{B}}^\mu \Hat{\mathcal{B}}^\nu{+}\kappa'_H \varepsilon^{\mu \nu \alpha\beta}u_\alpha\Hat{\mathcal{B}}_\beta) \nabla_\nu (\K T) {+}\mathcal{O}(\tau^2). \\
\end{split}
\end{equation}
The coefficients with subscript $\mathcal{B}$ describe anisotropies arising from the fact that the electrons move more easily along the magnetic field lines than orthogonally to it \cite{Aguilera:2007xk}. The coefficients with subscript $H$ describe Hall-related effects. Note that additional tensor structures such as $u^\mu \Hat{\mathcal{B}}^\nu$ are assumed negligible (if not outright zero) because they violate parity.
As before, one can use field redefinitions to fix the values of $\sigma_2$, $\sigma_4$, $\kappa_2$, and $\kappa_4$ in a way to make the theory causal. In particular, we need to adjust them so that the roots of the characteristic determinant
\begin{equation}
\det\left\{
\begin{bmatrix}
\sigma_1 g^{\mu \nu}{+}\sigma_2 u^\mu u^\nu{+}\sigma_\mathcal{B} \Hat{\mathcal{B}}^\mu \Hat{\mathcal{B}}^\nu{+}\sigma_H \varepsilon^{\mu \nu \alpha\beta}u_\alpha\Hat{\mathcal{B}}_\beta & &  \sigma_3 g^{\mu \nu}{+}\sigma_4 u^\mu u^\nu{+}\sigma'_\mathcal{B} \Hat{\mathcal{B}}^\mu \Hat{\mathcal{B}}^\nu{+}\sigma'_H \varepsilon^{\mu \nu \alpha\beta}u_\alpha\Hat{\mathcal{B}}_\beta \\  
\kappa_1 g^{\mu \nu}{+}\kappa_2 u^\mu u^\nu{+}\kappa_\mathcal{B} \Hat{\mathcal{B}}^\mu \Hat{\mathcal{B}}^\nu{+}\kappa_H \varepsilon^{\mu \nu \alpha\beta}u_\alpha\Hat{\mathcal{B}}_\beta & & \kappa_3 g^{\mu \nu}{+}\kappa_4 u^\mu u^\nu{+}\kappa'_\mathcal{B} \Hat{\mathcal{B}}^\mu \Hat{\mathcal{B}}^\nu{+}\kappa'_H \varepsilon^{\mu \nu \alpha\beta}u_\alpha\Hat{\mathcal{B}}_\beta\\
\end{bmatrix}\xi_\mu \xi_\nu\right\}=0 
\end{equation}
are such that the ``frequency'' $u^\mu \xi_\mu$ always falls within the interval $[-1,1]$ for all real ``wavevectors'' $(g^{\mu \nu}{+}u^\mu u^\nu)\xi_\nu$ with length 1 (note that all Hall contributions to the determinant vanish since $\varepsilon^{\mu \nu\alpha\beta}\xi_\mu \xi_\nu=0$). For illustration, let us consider the extreme case where the magnetic field is so large that all transport occurs along the field lines, so that $\sigma_1{=}\sigma_3{=}\kappa_1{=}\kappa_3{=}0$ \cite{Chandra:2015iza}. Then, we can just set $\sigma_2{=}{-}\sigma_\mathcal{B}$, $\sigma_4{=}{-}\sigma'_\mathcal{B}$, $\kappa_2{=}{-}\kappa_\mathcal{B}$, and $\kappa_4{=}{-}\kappa'_\mathcal{B}$, and the determinant becomes
\begin{equation}
\left[(-u^\mu u^\mu +\Hat{\mathcal{B}}^\mu\Hat{\mathcal{B}}^\nu)\xi_\mu \xi_\nu\right]^2 \det
\begin{bmatrix}
\sigma_\mathcal{B} & &  \sigma'_\mathcal{B} \\  
\kappa_\mathcal{B} & & \kappa'_\mathcal{B}\\
\end{bmatrix}=0 \, .
\end{equation}
Assuming that $\sigma_\mathcal{B}\kappa'_\mathcal{B}-\sigma'_\mathcal{B}\kappa_\mathcal{B}\neq 0$, we find that $(u^\mu \xi_\mu)^2=(\Hat{\mathcal{B}}^\mu \xi_\mu)^2 \leq 1$, i.e. causality.

\label{lastpage}

\end{document}